\documentclass[aps,pra,10pt,tightenlines,longbibliography,notitlepage,amsmath,twocolumn,floatfix,superscriptaddress]{revtex4-1}
\usepackage[dvipsnames]{xcolor}
\usepackage[colorlinks=true,hypertexnames=false,bookmarks=false,linkcolor=NavyBlue,urlcolor=NavyBlue,citecolor=NavyBlue,breaklinks]{hyperref}
\usepackage{amsmath,amsfonts,amssymb,amsthm}
\usepackage{mathtools}
\usepackage{thmtools}
\renewcommand\thmcontinues[1]{\textbf{continued}}
\usepackage{tabularx}
\usepackage{graphicx}
\usepackage{physics}
\usepackage[T1]{fontenc}
\usepackage{newtxtext,inconsolata}
\usepackage{enumitem}
\usepackage{setspace}
\usepackage{verbatimbox}
\usepackage{centernot}

\newtheorem{theorem}{Theorem}

\newtheorem{lemma}{Lemma}

\newcommand{\bk}[1]{\qty(#1)}
\newcommand{\Bk}[1]{\qty[#1]}
\newcommand{\BK}[1]{\qty{#1}}
\newcommand{\mc}[1]{\mathcal{#1}}
\newcommand{\mb}[1]{\mathbb{#1}}
\newcommand{\ms}[1]{\mathsf{#1}}

\renewcommand{\norm}[1]{\left\Vert #1 \right\Vert}

\renewcommand{\bra}[1]{\langle #1 |}
\renewcommand{\ket}[1]{| #1 \rangle}

\newcommand{\Hel}{\ms{Hel}}

\newcommand{\el}{\omega}

\newcommand{\error}{\mathrm{error}}

\newcommand{\HN}{\ms{HN}}

\newcommand{\cgraphic}[2]{\centerline{\includegraphics[width=#1\textwidth]{#2}}}
\newcommand{\fig}[3]{
\begin{figure}[htbp!]
\cgraphic{#1}{#2}
\caption{\label{#2}#3}
\end{figure}
}

\begin{document}

\title{Approaching the ultimate limit of quantum multiparameter estimation
by many-body physics}

\author{Mankei Tsang}
\email{mankei@nus.edu.sg}
\homepage{https://blog.nus.edu.sg/mankei/}
\affiliation{Department of Electrical and Computer Engineering,
  National University of Singapore, 4 Engineering Drive 3, Singapore
  117583}

\affiliation{Department of Physics, National University of Singapore,
  2 Science Drive 3, Singapore 117551}

\date{\today}

\begin{abstract}
  I propose a physical measurement scheme on multiple independent and
  identically distributed quantum objects to approach the
  Holevo--Nagaoka bound for quantum multiparameter estimation.  The
  scheme entails a physical interaction of the objects with bosonic
  ancillas, followed by a general-dyne measurement of the ancillas.
  The proposal offers a more concrete description of the experimental
  setup needed to achieve the ultimate precision limit set by the
  bound.
\end{abstract}

\maketitle
The estimation of multiple parameters of quantum objects, such as
optical fields or atomic spins, is a fundamental task in quantum
metrology, with important applications in sensing and imaging
\cite{helstrom,degen17,review_cp,lvovsky26,liu19,demkowicz20,pezze25}. Many
Cram\'er--Rao-type bounds have been proposed as fundamental limits to
the estimation precision \cite{helstrom,demkowicz20}. Among the
bounds, the so-called Holevo--Nagaoka bound for $N$ independent and
identically distributed (IID) objects
\cite{holevo_aspect,nagaoka89,hayashi24}, also called the Holevo
Cram\'er--Rao bound, can be regarded as the ultimate quantum limit that
is, in theory, achievable asymptotically by a collective measurement.
The most general measurement method to achieve the bound is called the
two-step method \cite{kahn09,gill_guta,yamagata13,yang19}: First, some
negligible number of objects are measured to give a preliminary
estimate $\check\theta$ of the parameter $\theta$. Second, a set of
collective observables $X^{(M)}$ of the remaining $M\approx N$ objects
are derived from $\check\theta$, such that a certain measurement of
$X^{(M)}$ can be shown to asymptotically attain the Holevo--Nagaoka
bound. While both steps are nontrivial, the second step is arguably
more difficult, since there is no known physical setup in general that
can measure a set of possibly incompatible observables $X^{(M)}$ to
the desired precision. In the limit of $M \to \infty$, $X^{(M)}$ can
be shown to approach bosonic quadrature operators by the quantum
central limit theorem (QCLT) \cite{ohya,verbeure}, but the limit is
mathematical---$X^{(M)}$ for any finite $M$ may not be bosonic
quadratures that can be easily measured in practice. ``Almost
experimentally infeasible'' is how one recent work describes the
measurement \cite{zhou26}.

In recent years, significant effort has been devoted to finding good
measurements for quantum multiparameter estimation---see, for example,
\cite{matsumoto02,twc,vidrighin,genoni16,tnl,review_cp,lvovsky26,parniak18,hou18,demkowicz20,albarelli20a,ansari21,chen22,chen24,pelecanos25,zhou26}
and references therein---but the proposals so far have been limited to
special cases, such as pure states or low-rank states
\cite{matsumoto02,zhou26}, or they involve nontrivial numerical
methods, unclear experimental implementations, or complicated quantum
circuits. The purpose of this paper is to outline a general physical
scheme that can measure the collective observables $X^{(M)}$ optimally
in the second step of the two-step method. The idea is simple: since
$X^{(M)}$ are physical observables, a Hamiltonian that couples them to
bona-fide bosonic ancillas, such as optical modes, is physical as
well. The design of the interaction can be facilitated by the QCLT,
which enables one to approximate $X^{(M)}$ by bosonic quadratures
$X^{(\infty)}$ for large $M$. It is then straightforward to write down
a Hamiltonian that can faithfully transduce the statistics of
$X^{(\infty)}$ to the ancilla quadratures. After the interaction, the
ancilla quadratures, being physically bosonic, can be measured more
easily.

Let the Hilbert space for a quantum object be $\mc H$ and its state be
the density operator $\rho(\theta)$ on $\mc H$. Let $\rho(\theta)$ be
a function of the unknown parameter $\theta \in \Theta$ in a parameter
space $\Theta$. The state of $N$ IID copies is
$\rho(\theta)^{\otimes N}$ on $\mc H^{\otimes N}$. Let
$\beta:\Theta \to \mb R^n$ be a vectoral parameter of interest to be
estimated. Let the probability measure for a measurement of the $N$
objects be $P_N(\cdot|\theta)$ and an estimator be
$\check\beta^{(N)}(\el)$ in terms of the outcome $\el$.  Define the
estimation error matrix as
$V^{(N)} \equiv \int \bk{\check\beta^{(N)}-\beta}
\bk{\check\beta^{(N)}-\beta}^\top dP_N$, where $\beta$ and
$\check\beta^{(N)}$ are assumed to be column vectors, $\top$ denotes
the transpose, and all quantities are implicitly assumed to be
evaluated at the same $\theta$ called the true parameter. Define the
average error as
$\error_N(\theta) \equiv \trace\Bk{W V^{(N)}(\theta)}$, where $\trace$
denotes the trace and $W$ is a positive-definite matrix that defines
the weights in the averaging.  Under certain asymptotic conditions on
the estimator, the asymptotic normalized error can be bounded as
\cite{yang19}
\begin{align}
\lim_{N\to\infty} N \error_N(\theta) \ge  \HN(\theta)
\equiv \inf_{\delta \in \mc D} C(\delta),
\label{HN}
\\
 C(\delta) \equiv \trace\Bk{W \Re\Gamma(\delta)}
 + \trace\Bk{\abs{\sqrt{W}\Im\Gamma(\delta)\sqrt{W}}},
\label{HN2}
\\
\Gamma_{jk}(\delta) \equiv \trace\bk{\delta_j\delta_k \rho},
\end{align}
where $\HN$ denotes the Holevo--Nagaoka bound,
$|A| \equiv \sqrt{A^\dagger A}$, $\dagger$ denotes the conjugate
transpose, $\Re$ and $\Im$ denote the entry-wise real and
imaginary parts, respectively,
$\delta \equiv \mqty(\delta_1&\dots&\delta_n)^\top \in \mc D$ is
called a (vectoral) influence operator, and $\mc D$ denotes the set of
all influence operators \cite{semi_prx}. Each $\delta_j$ is a
self-adjoint operator on $\mc H$ that satisfies
$\trace(\delta_j\rho) = 0$ and
$\trace(\delta_j\dot\phi\rho) = \dot\phi \beta_j$ for any parametric
submodel $\phi \in \Phi(\theta)$. $\phi$ and $\dot\phi$ are defined as
follows: Each $\phi$ is a curve in the parameter space
$\phi:(-\epsilon,\epsilon) \to \Theta$ for some $\epsilon > 0$ that
satisfies $\phi(0) = \theta$, while $\dot\phi$ denotes its directional
derivative at the true $\theta$ such that
$\dot\phi f \equiv \dv*{f(\phi(t))}{t}|_{t=0}$. (See
\cite[Sec.~III]{qlmoment_pra2} for the precise definition of the set
of parametric submodels $\Phi(\theta)$.)

A looser bound called the generalized Helstrom bound is defined as
\cite{semi_prx}
\begin{align}
\Hel &\equiv \inf_{\delta \in \mc D} \trace\Bk{W \Re\Gamma(\delta)},
\label{helstrom}
\end{align}
which coincides with $\HN$ for a scalar $\beta$. For many problems, an
influence operator $\delta^\Hel$ that achieves the infimum in
Eq.~(\ref{helstrom}) exists and is called the efficient influence.
$\Hel$ coincides with Helstrom's version of the quantum Cram\'er--Rao
bound \cite{helstrom} when the parameter space $\Theta$ is
finite-dimensional \cite{fujiwara05,semi_prx}. It is known that
\cite[Theorem~9]{semi_prx} (see also \cite{carollo19,*carollo20})
\begin{align}
\Hel &\le \HN \le C(\delta^\Hel) \le 2\Hel.
\label{bounds}
\end{align}
It is remarkable that $\HN/\Hel \le 2$ and $\HN$ is asymptotically
achievable for any number $n$ of parameters in $\beta$, even though
one may suspect that the issue of measurement incompatibility would
make the estimation problem harder for increasing $n$.

I now briefly review the two-step method to approach the
Holevo--Nagaoka bound; see \cite{gill_guta,yamagata13,yang19} for
precise details. Assume that an optimal influence operator $\delta$
achieving the infimum in Eq.~(\ref{HN}) exists. In the first step that
I will call the acquisition step, some negligible number of objects
are measured to obtain a preliminary estimate $\check\theta$ of
$\theta$, analogous to the initial acquisition stage of a phase-locked
loop \cite{vantrees2,wiseman_milburn}. $\check\theta$ should be
accurate enough so that one can find a set of observables closely
approximating
\begin{align}
X(\theta) &\equiv \delta(\theta) + \Bk{\beta(\theta)-c} I
\end{align}
for some $c \in \mb R^n$, where $I$ is the identity operator, and also
a matrix closely approximating $\Gamma[\delta(\theta)]$. For example,
under the condition of quantum local asymptotic normality, it can be
shown that
$\delta(\check\theta)\approx \delta(\theta) +
[\beta(\theta)-\beta(\check\theta)] I$ and
$\Gamma[\delta(\check\theta)]$ at $\rho(\check\theta)$ is close to
$\Gamma[\delta(\theta)]$ at $\rho(\theta)$, so one can set
$X = \delta(\check\theta)$ with $c = \beta(\check\theta)$ and
$\Gamma = \Gamma[\delta(\check\theta)]$. Define the collective form of
$X_j$ for the remaining $M \approx N$ objects as
\begin{align}
X_j^{(M)} &\equiv \frac{1}{\sqrt{M}} \sum_{k=1}^M I^{\otimes(k-1)}\otimes
 X_j \otimes I^{\otimes(M-k)}.
\label{XM}
\end{align}
By the QCLT \cite{ohya,verbeure}, the quantum state with respect to
$X^{(M)}-\sqrt{M}[\beta(\theta)-c]I$ converges to a zero-mean Gaussian
state with covariance matrix $\Re\Gamma(\delta)$ in the limit of
$M \to \infty$. Moreover,
\begin{align}
\Bk{X_j^{(M)},X_k^{(M)}} &\to 
\Bk{X_j^{(\infty)},X_k^{(\infty)}} = \Bk{2i \Im \Gamma_{jk}(\delta)} I,
\label{ergodic}
\end{align}
such that $X^{(\infty)}$ become bosonic quadrature operators.  Another
name for the QCLT is bosonization, a widely used method to simplify
many-body models by approximating collective observables as bosonic
quadratures \cite{verbeure,larson}.  In the second step of the
two-step method that I will call the Gaussian measurement step,
$X^{(M)}$ can be measured together with some bosonic ancilla
quadratures, such that the outcomes are well approximated by Gaussian
random variables with mean vector $\sqrt{M}[\beta(\theta)-c]$. An
appropriate Gaussian ancilla state can lead to the optimal asymptotic
error $N \error_N(\theta) \to C(\delta) = \HN$
\cite{gill_guta,demkowicz20}.

Assuming $W = I$, a more straightforward approach to the Gaussian
measurement step is to split the $M$ objects into $n$ batches and use
each batch to estimate only one $\beta_j$. By measuring
$X_j=\delta_j^\Hel + (\beta_j-c_j) I$ of each object in the $j$th
batch, this separable measurement achieves an error of
$(M/n)V_{jj}^{(M/n)} \to \trace[(\delta^\Hel_j)^2 \rho]$ for that
$\beta_j$, so the average error would be
$N \error_N(\theta) \to n \Hel$.  For a low $n$, this error may be
acceptable, but for $n \gg 2$, a more nontrivial measurement is
necessary to achieve the much lower Holevo--Nagaoka bound
$\HN \ll n\Hel$.

It is usually easier to find the efficient influence $\delta^\Hel$
that achieves the generalized Helstrom bound $\Hel$ than to find the
$\delta$ that achieves the Holevo--Nagaoka bound $\HN$.  Once
$\delta^\Hel$ is found from the acquisition step, one can follow the
same Gaussian measurement step mentioned earlier with
$X = \delta^\Hel + (\beta-c) I$ to achieve an asymptotic normalized
error $C(\delta^\Hel)$. Eq.~(\ref{bounds}) implies that
$\HN \le C(\delta^\Hel) \le 2\Hel \le 2 \HN$, so $\delta^\Hel$ is
adequate if the difference by at most a factor of 2 is acceptable.
$\delta^\Hel$ is always a real linear combination of the so-called
score operators $\{S(\dot\phi):\phi \in \Phi(\theta)\}$, also called
symmetric logarithmic derivatives, defined by
$\dot\phi \rho = [\rho S(\dot\phi)+S(\dot\phi)\rho]/2$
\cite{semi_prx,fujiwara05}.  One example is thermal-light sensing and
imaging \cite{helstrom,tnl,review_cp,lvovsky26}.  Given a thermal
source and any passive linear optics, the quantum state of the $m$
received optical modes has a zero-mean Gaussian Glauber--Sudarshan
representation with a covariance matrix given by the optical mutual
coherence matrix; let this matrix be
$\Upsilon:\Theta \to \mb C^{m\times m}$. The scores all have the
quadratic form
$S(\dot\phi) = \sum_{\mu,\nu} R_{\mu\nu}(\dot\phi) a_\mu^\dagger a_\nu
- \trace[R(\dot\phi) \Upsilon] I$ \cite[Eq.~(6.14),
p.~283]{helstrom}, where $\BK{a_\mu:\mu \in \Omega}$ are the
annihilation operators of the optical modes with mode indices $\Omega$
and the Hermitian matrix $R(\dot\phi) \in \mb C^{m\times m}$ is
determined by
$\dot\phi\Upsilon = [\Upsilon R(\dot\phi) \bk{I + \Upsilon} + \bk{I +
  \Upsilon} R(\dot\phi) \Upsilon]/2$ \cite[Eq.~(6.19),
p.~284]{helstrom}.  It follows that $\delta^\Hel$ also has the same
quadratic form with respect to $\BK{a_\mu}$, that is,
\begin{align}
\delta^\Hel_j &= 
\sum_{\mu,\nu} D_{j\mu\nu} a_\mu^\dagger a_\nu -\trace\bk{D_j\Upsilon} I
\label{effinf_thermal}
\end{align}
for some Hermitian $D_j \in \mb C^{m\times m}$.  In particular, for
diffraction-limited imaging of incoherent sources
\cite{tnl,review_cp}, $\Upsilon$ can be assumed to be real for all
$\theta$.  Then, using the technique in \cite{miyazaki22},
Appendix~\ref{app_imaging} shows that $\Im\Gamma(\delta^\Hel) = 0$ for
all $\theta$, meaning that $\Hel = \HN = C(\delta^\Hel)$ and
$\delta^\Hel$ also achieves the infimum in Eq.~(\ref{HN}) for
$\HN$. Another example is the estimation of the expected values
$\beta_j(\theta) = \trace\Bk{b_j\rho(\theta)}$ of a set of observables
$b = \mqty(b_1&\dots&b_n)^\top$.  Assuming the nonparametric model,
where $\rho(\theta) = \theta$ and $\Theta$ is the set of all density
operators, it can be proved that $\delta^\Hel = b - \beta I$
\cite{semi_prx}, so one can set
\begin{align}
X &= \delta^\Hel + \beta I = b.
\label{nonparam}
\end{align}
For the nonparametric model, it can in fact be proved that there is
only one influence operator $\delta = \delta^\Hel$ in $\mc D$
\cite{semi_prx} (in terms of the equivalence classes described in
\cite{holevo77}), so $\delta^\Hel$ must also be optimal for $\HN$.

While the collective observables $X^{(M)}$ tend to bosonic quadratures
mathematically, $X^{(M)}$ for any finite $M$ may not be bosonic
quadratures physically. In the examples above, one finds that the
single-object observables $X$ may be observables of nonbosonic
objects, or even if the objects are bosonic, $X$ may not be their
quadratures. I now outline an indirect measurement of $X^{(M)}$ that
is the chief result of this work. First, I transform $X^{(M)}$ to a
more convenient set of observables. Since $\Im\Gamma(\delta)$ is a
real and antisymmetric matrix, there exists an orthogonal matrix
$O \in \mb R^{n\times n}$ such that \cite[Corollary~2.5.11]{horn}
\begin{align}
O^\top \Im\Gamma(\delta) O &= 
\Bk{\bigoplus_{j=1}^r \mqty(0 & \nu_j\\ -\nu_j & 0) }\oplus 0,
\label{ImG}
\end{align}
where $\BK{\nu_j}$ are all positive and the last $0$ denotes a zero
matrix of appropriate size ($(n-2r)\times(n-2r)$ here). Then there
exists an invertible matrix $L \in \mb R^{n\times n}$ such that
\begin{align}
X^{(M)} &= L Y^{(M)},
\label{L}
&
Y^{(M)} &= 
\Bk{\bigoplus_{j=1}^r \mqty(q_j^{(M)}\\ p_j^{(M)})}
\oplus \mqty(x_1^{(M)} \\ \vdots \\ x_{n-2r}^{(M)}),
\end{align}
where $Y^{(M)}$ satisfy the asymptotic commutation relations
\begin{align}
\Bk{q_j^{(M)},p_k^{(M)}} &\to i\delta_{jk} I,
&
\Bk{Y_j^{(M)},Y_k^{(M)}} &\to 0\ \textrm{otherwise}.
\label{Ycompat}
\end{align}
Each $Y_j^{(M)}$ is in the same collective form as Eq.~(\ref{XM}) in
terms of a single-object observable $Y_j = \sum_k (L^{-1})_{jk} X_k$.
Eqs.~(\ref{ImG})--(\ref{Ycompat}) mean that $X^{(M)}$ can always be
transformed to a set of observables $Y^{(M)}$ that comprise $r$ pairs
of canonical quadratures
$\bigoplus_{j=1}^r \mqty(q_j^{(\infty)} & p_j^{(\infty)})^\top$ and
$n-2r$ classical variables
$\mqty(x_1^{(\infty)} & \dots & x_{n-2r}^{(\infty)})^\top$ in the
asymptotic limit. Each $q_j^{(\infty)}$ or $p_j^{(\infty)}$ is
incompatible only with its conjugate observable and commutes with all
the other entries of $Y^{(\infty)}$, while each $x_j^{(\infty)}$
commutes with all the other entries of $Y^{(\infty)}$ and is thus
effectively classical. This high degree of asymptotic compatibility
within $Y^{(M)}$ is the fundamental reason that the Holevo--Nagaoka
bound is achievable, even though the single-object observables $X$ may
seem highly incompatible with one another and difficult to measure
simultaneously.  An optimal measurement of $X^{(\infty)}$ satisfying
$N\error_N(\theta) \to C(\delta)$ can be implemented by an appropriate
measurement of $Y^{(\infty)}$ \cite[Sec.~3.2]{demkowicz20}.
Remarkably, for diffraction-limited incoherent imaging, where
$\Im\Gamma(\delta) = 0$ for all $\theta$, all entries of $X^{(M)}$ are
asymptotically compatible with one another and one can simply set
$X^{(M)} = Y^{(M)} = x^{(M)}$.

Next, I consider an indirect measurement of
$\bigoplus_{j} \mqty(q_j^{(M)}& p_j^{(M)})^\top$. Assume the
Hamiltonian
\begin{align}
H_1^{(M)} &= \kappa \sum_{j=1}^r \bk{q_j^{(M)}  p_j' - p_j^{(M)} q_j'},
\label{BS}
\end{align}
where $\kappa > 0$ and $(q_j',p_j')$ are canonical quadratures of a
bona-fide bosonic ancilla mode, such that each ancilla interacts with
the collective observables $(q_j^{(M)},p_j^{(M)})$ of the $M$
objects. Approximating $H_1^{(M)}$ by $H_1^{(\infty)}$---which is a
``beam-splitter'' Hamiltonian---and taking $\kappa t = \pi/2$ (with
$\hbar = 1$), I find, in the Heisenberg picture,
\begin{align}
q_j'(t) &= q_j^{(\infty)},
&
p_j'(t) &= p_j^{(\infty)}.
\label{swap}
\end{align}
Now one can measure the ancilla quadratures at time $t$; they have the
same Gaussian statistics as
$\bigoplus_j \mqty (q_j^{(\infty)} & p_j^{(\infty)})^\top$ of the
original objects.  If the ancillas are optical modes, the optimal
measurement is a general-dyne measurement \cite{genoni16}, which
requires $r$ additional optical ancilla modes. Ref.~\cite{sm} proves
that the required state of these additional ancillas is always
physical, and if $W = I$, the state is the vacuum state, meaning that
the general-dyne measurement is simply heterodyne.

Lastly, I consider an indirect measurement of
$x^{(M)} = \mqty(x_1^{(M)}&\dots&x_{n-2r}^{(M)})^\top$. If a set of
conjugate observables
$y^{(M)} = \mqty(y_1^{(M)}&\dots& y_{n-2r}^{(M)})^\top$ can be found
for the quantum objects such that
$[x_j^{(M)},y_k^{(M)}] \to i \delta_{jk} I$,
$[y_j^{(M)},y_k^{(M)}] \to 0$, and $[y_j^{(M)},Y_k^{(M)}] \to 0$ with
all the other entries of $Y^{(M)}$, then $(x_j^{(M)},y_j^{(M)})$
approach canonical quadratures of a bosonic mode in the same manner as
$(q_j^{(M)},p_j^{(M)})$, and a beam-splitter Hamiltonian in the form
of Eq.~(\ref{BS}) can still be used to transfer the statistics of
$x^{(M)}$ to some ancilla quadratures $q''$.  Alternatively, consider
the Hamiltonian
\begin{align}
H^{(M)} &= H_1^{(M)} + H_2^{(M)},
&
H_2^{(M)} &= \gamma \sum_{j=1}^{n-2r} x_j^{(M)}  p_j'',
\label{disp}
\end{align}
where $\gamma \in \mb R$ and $(q_j'',p_j'')$ are again canonical
quadratures of a bosonic ancilla mode. Approximating $H^{(M)}$ as
$H^{(\infty)}$, I find
\begin{align}
q_j''(t) &=  q_j'' + \gamma t x_j^{(\infty)},
\end{align}
such that each ancilla quadrature $q_j''(t)$ is displaced by the
object observable $x_j^{(\infty)}$. This type of interaction has been
extensively studied in the context of quantum-nondemolition (QND)
measurements \cite{braginsky,grangier98}.  As long as the initial mean
of $q''$ is known and the initial covariance matrix of $q''$ is
negligible relative to $(\gamma t)^2$ times the covariance matrix of
$x^{(\infty)}$, one can measure $q''(t)$ of the ancillas as an
indirect measurement of $x^{(\infty)}$ with the desired mean and
covariances.

Fig.~\ref{manybody_scheme3} illustrates the complete scheme to measure
$Y^{(M)}$.  To obtain the desired measurement of $X^{(M)}$, simply
multiply the column vector of outcomes by the $L$ matrix in
Eq.~(\ref{L}). $n$ ancilla modes are needed in total---it is good news
that this number does not increase with the number of objects $M$.
The coupling coefficients $\kappa$ and $\gamma$ in Eqs.~(\ref{BS}) and
(\ref{disp}) also scale favorably with $M$: Let the physical coupling
coefficients for one object ($M = 1$) be $\kappa = \kappa_1$ and
$\gamma = \gamma_1$; the effective coupling coefficients for $M$
objects are then given by $\kappa = \sqrt{M}\kappa_1$ and
$\gamma = \sqrt{M}\gamma_1$.

\fig{0.45}{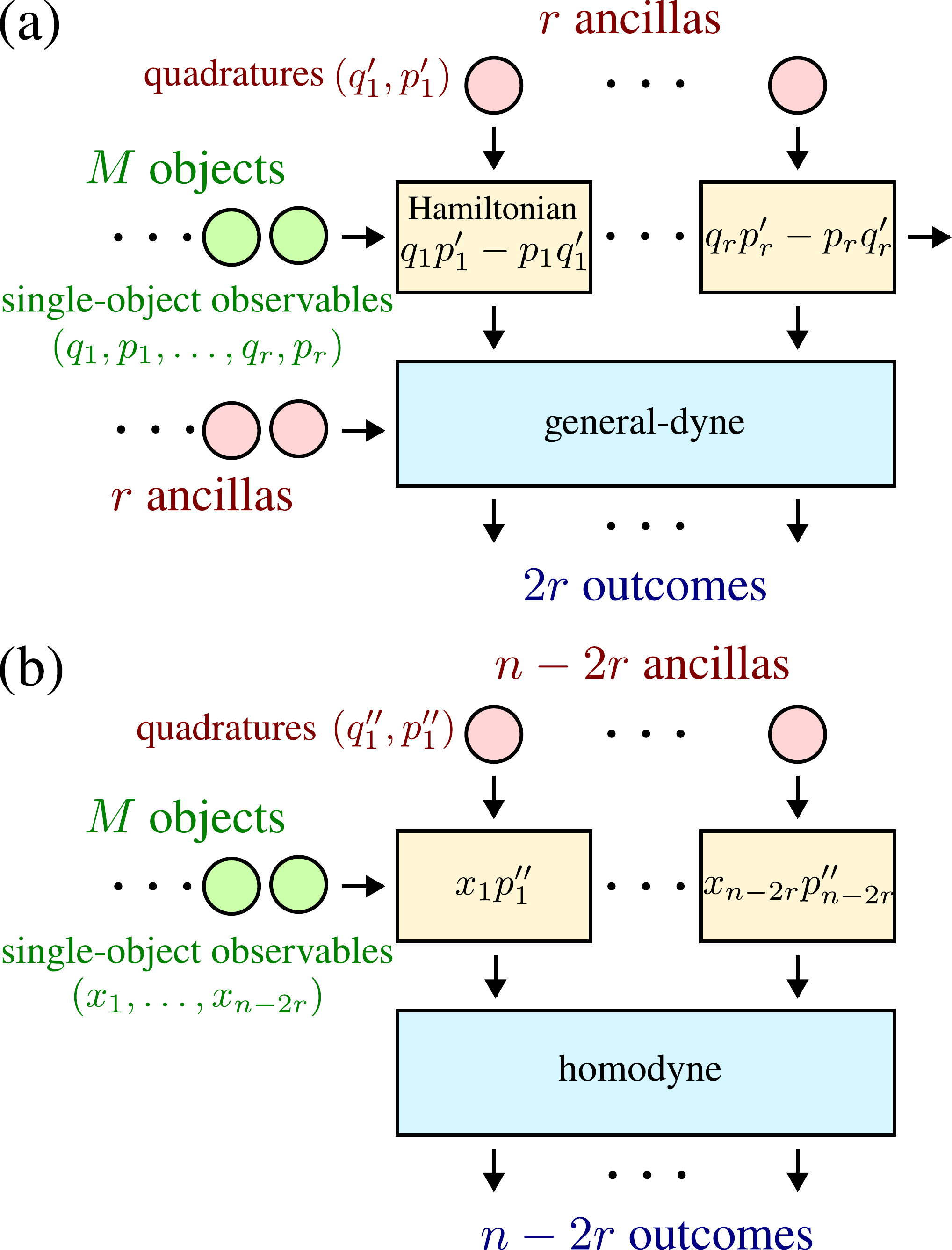}{Schematics of the proposed
    measurement.  $Y = (q_1,p_1,\dots,x_1,\dots)$ are the
    single-object observables derived from $Y = L^{-1} X$ using the
    $L$ matrix in Eq.~(\ref{L}).  Each object interacts with (a) an
    ancilla with quadratures $(q_j',p_j')$ as per the Hamiltonian
    $\propto q_jp_j'-p_j q_j'$ and (b) an ancilla with quadratures
    $(q_j'',p_j'')$ as per the Hamiltonian $\propto x_j p_j''$, such
    that the total Hamiltonian for $M$ objects is given by
    Eqs.~(\ref{BS}) and (\ref{disp}). Assuming that the ancillas are
    optical modes, a general-dyne measurement is then performed on the
    $r$ ancillas in (a) with $r$ additional ancillas, while a homodyne
    measurement is performed on the $n-2r$ ancillas in (b).}

The essential feature of the proposed scheme is that the Hamiltonians
are expressed in terms of physical observables
$Y^{(M)} = L^{-1} X^{(M)}$ of the quantum objects, so that it becomes
clearer how the dynamics can be implemented in reality. For example,
if each $Y_j$ is a spin observable, then $Y_j^{(M)}$ is a collective
spin observable of the $M$ objects, the Hamiltonians resemble those of
the Dicke and Tavis--Cummings models, and the QCLT-based approximation
resembles the Holstein--Primakoff approximation
\cite{verbeure,larson,hammerer10}. For the thermal-light example,
$X^{(M)}$ and $Y^{(M)}$ in terms of Eq.~(\ref{effinf_thermal}) are
quadratic with respect to optical annihilation and creation operators,
meaning that the Hamiltonian $H^{(M)}$ is cubic with respect to
bosonic field operators. This cubic Hamiltonian implies that, in the
context of optics, only a three-wave-mixing nonlinearity (or
multi-wave mixing with strong pumps) \cite{drummond,firstenberg16} is
needed to implement the interactions---the order of the nonlinearity
fortunately does not increase with
$M$. Appendix~\ref{app_imaging} further describes some concrete
  implementations of the measurement for incoherent imaging.

The $M \to \infty$ approximation is crucial for the design of the
Hamiltonians and the proof that the dynamics work as desired, free of
any unnecessary measurement-backaction noise \cite{braginsky}.
Another viewpoint is to regard the approximation as a linearization of
the equations of motion and error propagation; the replacement of
$H^{(M)}$ by the bosonic quadratic $H^{(\infty)}$ is simply a more
rigorous approach based on the QCLT at the Hamiltonian level.
Appendix~\ref{app_spin} works out a simple spin example to show that
linearization gives the same result. The success of the bosonization
method in many-body physics \cite{verbeure,larson,hammerer10} suggests
that it should work well for a large enough $M$.

The proposed scheme, though conceptually simple and scalable, is
imperfect in many ways. First of all, the scheme does not implement
the acquisition step. For an infinite-dimensional $\mc H$, it is not
even known if the acquisition step can always achieve the required
accuracy, or under what condition it can do so.  I stress, however,
that many problems may not require the step
\cite{demkowicz20,tnl,review_cp}; for example, in astronomical
imaging, ample prior information about the stellar objects may already
exist.  For problems that do require acquisition, a viable method may
be to use the same setup for both steps with adaptive
measurement-based feedback \cite{wiseman_milburn}. Second, it is
unclear how large $M$ should be to make the QCLT-based approximation
satisfactory, although the achievability of the Holevo--Nagaoka bound
is proved using the same approximation in the first place and thus
shares the same question. Third, the scheme does not rule out more
clever methods that may work better for finite $M$ and special cases
\cite{matsumoto02,twc,vidrighin,genoni16,tnl,review_cp,lvovsky26,parniak18,hou18,demkowicz20,albarelli20a,ansari21,chen22,chen24,pelecanos25,zhou26}.
Finally, even if all the approximations are valid, the proposed
Hamiltonians may still be challenging to implement experimentally. The
experimental effort may nonetheless be rewarding: compared with the
easier separable measurements \cite{gill_massar,conlon21}, the
improvement due to collective measurements can be substantial when
$\beta$ contains many parameters \cite{belliardo21}.

An insightful discussion with Rafa{\l} Demkowicz-Dobrza\'nski is
gratefully acknowledged.


\begin{thebibliography}{53}%
\makeatletter
\providecommand \@ifxundefined [1]{%
 \@ifx{#1\undefined}
}%
\providecommand \@ifnum [1]{%
 \ifnum #1\expandafter \@firstoftwo
 \else \expandafter \@secondoftwo
 \fi
}%
\providecommand \@ifx [1]{%
 \ifx #1\expandafter \@firstoftwo
 \else \expandafter \@secondoftwo
 \fi
}%
\providecommand \natexlab [1]{#1}%
\providecommand \enquote  [1]{``#1''}%
\providecommand \bibnamefont  [1]{#1}%
\providecommand \bibfnamefont [1]{#1}%
\providecommand \citenamefont [1]{#1}%
\providecommand \href@noop [0]{\@secondoftwo}%
\providecommand \href [0]{\begingroup \@sanitize@url \@href}%
\providecommand \@href[1]{\@@startlink{#1}\@@href}%
\providecommand \@@href[1]{\endgroup#1\@@endlink}%
\providecommand \@sanitize@url [0]{\catcode `\\12\catcode `\$12\catcode
  `\&12\catcode `\#12\catcode `\^12\catcode `\_12\catcode `\%12\relax}%
\providecommand \@@startlink[1]{}%
\providecommand \@@endlink[0]{}%
\providecommand \url  [0]{\begingroup\@sanitize@url \@url }%
\providecommand \@url [1]{\endgroup\@href {#1}{\urlprefix }}%
\providecommand \urlprefix  [0]{URL }%
\providecommand \Eprint [0]{\href }%
\providecommand \doibase [0]{http://dx.doi.org/}%
\providecommand \selectlanguage [0]{\@gobble}%
\providecommand \bibinfo  [0]{\@secondoftwo}%
\providecommand \bibfield  [0]{\@secondoftwo}%
\providecommand \translation [1]{[#1]}%
\providecommand \BibitemOpen [0]{}%
\providecommand \bibitemStop [0]{}%
\providecommand \bibitemNoStop [0]{.\EOS\space}%
\providecommand \EOS [0]{\spacefactor3000\relax}%
\providecommand \BibitemShut  [1]{\csname bibitem#1\endcsname}%
\let\auto@bib@innerbib\@empty
\bibitem [{\citenamefont {Helstrom}(1976)}]{helstrom}%
  \BibitemOpen
  \bibfield  {author} {\bibinfo {author} {\bibfnamefont {Carl~W.}\ \bibnamefont
  {Helstrom}},\ }\href
  {http://www.sciencedirect.com/science/bookseries/00765392/123} {\emph
  {\bibinfo {title} {Quantum Detection and Estimation Theory}}}\ (\bibinfo
  {publisher} {Academic Press},\ \bibinfo {address} {New York},\ \bibinfo
  {year} {1976})\BibitemShut {NoStop}%
\bibitem [{\citenamefont {Degen}\ \emph {et~al.}(2017)\citenamefont {Degen},
  \citenamefont {Reinhard},\ and\ \citenamefont {Cappellaro}}]{degen17}%
  \BibitemOpen
  \bibfield  {author} {\bibinfo {author} {\bibfnamefont {C.~L.}\ \bibnamefont
  {Degen}}, \bibinfo {author} {\bibfnamefont {F.}~\bibnamefont {Reinhard}}, \
  and\ \bibinfo {author} {\bibfnamefont {P.}~\bibnamefont {Cappellaro}},\
  }\bibfield  {title} {\enquote {\bibinfo {title} {Quantum sensing},}\ }\href
  {\doibase 10.1103/RevModPhys.89.035002} {\bibfield  {journal} {\bibinfo
  {journal} {Reviews of Modern Physics}\ }\textbf {\bibinfo {volume} {89}},\
  \bibinfo {pages} {035002} (\bibinfo {year} {2017})}\BibitemShut {NoStop}%
\bibitem [{\citenamefont {Tsang}(2019)}]{review_cp}%
  \BibitemOpen
  \bibfield  {author} {\bibinfo {author} {\bibfnamefont {Mankei}\ \bibnamefont
  {Tsang}},\ }\bibfield  {title} {\enquote {\bibinfo {title} {Resolving
  starlight: A quantum perspective},}\ }\href {\doibase
  10.1080/00107514.2020.1736375} {\bibfield  {journal} {\bibinfo  {journal}
  {Contemporary Physics}\ }\textbf {\bibinfo {volume} {60}},\ \bibinfo {pages}
  {279–298} (\bibinfo {year} {2019})}\BibitemShut {NoStop}%
\bibitem [{\citenamefont {Lvovsky}\ \emph {et~al.}(2026)\citenamefont
  {Lvovsky}, \citenamefont {Grace}, \citenamefont {Guha}, \citenamefont
  {Tsang}, \citenamefont {Adesso},\ and\ \citenamefont {Treps}}]{lvovsky26}%
  \BibitemOpen
  \bibfield  {author} {\bibinfo {author} {\bibfnamefont {A.~I.}\ \bibnamefont
  {Lvovsky}}, \bibinfo {author} {\bibfnamefont {Michael~R.}\ \bibnamefont
  {Grace}}, \bibinfo {author} {\bibfnamefont {Saikat}\ \bibnamefont {Guha}},
  \bibinfo {author} {\bibfnamefont {Mankei}\ \bibnamefont {Tsang}}, \bibinfo
  {author} {\bibfnamefont {Gerardo}\ \bibnamefont {Adesso}}, \ and\ \bibinfo
  {author} {\bibfnamefont {Nicolas}\ \bibnamefont {Treps}},\ }\href {\doibase
  10.48550/arXiv.2605.10767} {\enquote {\bibinfo {title} {Passive optical
  superresolution at the quantum limit},}\ } (\bibinfo {year} {2026}),\ \Eprint
  {http://arxiv.org/abs/2605.10767} {arXiv:2605.10767 [quant-ph]} \BibitemShut
  {NoStop}%
\bibitem [{\citenamefont {Liu}\ \emph {et~al.}(2019)\citenamefont {Liu},
  \citenamefont {Yuan}, \citenamefont {Lu},\ and\ \citenamefont
  {Wang}}]{liu19}%
  \BibitemOpen
  \bibfield  {author} {\bibinfo {author} {\bibfnamefont {Jing}\ \bibnamefont
  {Liu}}, \bibinfo {author} {\bibfnamefont {Haidong}\ \bibnamefont {Yuan}},
  \bibinfo {author} {\bibfnamefont {Xiao-Ming}\ \bibnamefont {Lu}}, \ and\
  \bibinfo {author} {\bibfnamefont {Xiaoguang}\ \bibnamefont {Wang}},\
  }\bibfield  {title} {\enquote {\bibinfo {title} {Quantum {{Fisher}}
  information matrix and multiparameter estimation},}\ }\href {\doibase
  10.1088/1751-8121/ab5d4d} {\bibfield  {journal} {\bibinfo  {journal} {Journal
  of Physics A: Mathematical and Theoretical}\ }\textbf {\bibinfo {volume}
  {53}},\ \bibinfo {pages} {023001} (\bibinfo {year} {2019})}\BibitemShut
  {NoStop}%
\bibitem [{\citenamefont {{Demkowicz-Dobrzański}}\ \emph
  {et~al.}(2020)\citenamefont {{Demkowicz-Dobrzański}}, \citenamefont
  {Górecki},\ and\ \citenamefont {Guţă}}]{demkowicz20}%
  \BibitemOpen
  \bibfield  {author} {\bibinfo {author} {\bibfnamefont {Rafał}\ \bibnamefont
  {{Demkowicz-Dobrzański}}}, \bibinfo {author} {\bibfnamefont {Wojciech}\
  \bibnamefont {Górecki}}, \ and\ \bibinfo {author} {\bibfnamefont
  {Mădălin}\ \bibnamefont {Guţă}},\ }\bibfield  {title} {\enquote {\bibinfo
  {title} {Multi-parameter estimation beyond quantum {{Fisher}} information},}\
  }\href {\doibase 10.1088/1751-8121/ab8ef3} {\bibfield  {journal} {\bibinfo
  {journal} {Journal of Physics A: Mathematical and Theoretical}\ }\textbf
  {\bibinfo {volume} {53}},\ \bibinfo {pages} {363001} (\bibinfo {year}
  {2020})}\BibitemShut {NoStop}%
\bibitem [{\citenamefont {Pezzè}\ and\ \citenamefont
  {Smerzi}(2025)}]{pezze25}%
  \BibitemOpen
  \bibfield  {author} {\bibinfo {author} {\bibfnamefont {Luca}\ \bibnamefont
  {Pezzè}}\ and\ \bibinfo {author} {\bibfnamefont {Augusto}\ \bibnamefont
  {Smerzi}},\ }\href {\doibase 10.48550/arXiv.2502.17396} {\enquote {\bibinfo
  {title} {Advances in multiparameter quantum sensing and metrology},}\ }
  (\bibinfo {year} {2025}),\ \bibinfo {note} {comment: 12 pages, 4 figures,
  comments are welcome},\ \Eprint {http://arxiv.org/abs/2502.17396}
  {arXiv:2502.17396 [quant-ph]} \BibitemShut {NoStop}%
\bibitem [{\citenamefont {Holevo}(2011)}]{holevo_aspect}%
  \BibitemOpen
  \bibfield  {author} {\bibinfo {author} {\bibfnamefont {Alexander~S.}\
  \bibnamefont {Holevo}},\ }\href {\doibase 10.1007/978-88-7642-378-9} {\emph
  {\bibinfo {title} {Probabilistic and Statistical Aspects of Quantum
  Theory}}}\ (\bibinfo  {publisher} {Scuola Normale Superiore Pisa},\ \bibinfo
  {address} {Pisa, Italy},\ \bibinfo {year} {2011})\BibitemShut {NoStop}%
\bibitem [{\citenamefont {Nagaoka}(1989)}]{nagaoka89}%
  \BibitemOpen
  \bibfield  {author} {\bibinfo {author} {\bibfnamefont {Hiroshi}\ \bibnamefont
  {Nagaoka}},\ }\bibfield  {title} {\enquote {\bibinfo {title} {A new approach
  to {{Cramér--Rao}} bounds for quantum state estimation},}\ }\href@noop {}
  {\bibfield  {journal} {\bibinfo  {journal} {IEICE Technical Report}\ }\textbf
  {\bibinfo {volume} {IT 89-42}},\ \bibinfo {pages} {9–14} (\bibinfo {year}
  {1989})},\ \bibinfo {note} {reprinted in
  \href{http://dx.doi.org/10.1142/9789812563071\_0009}{\emph{Asymptotic Theory
  of Quantum Statistical Inference: Selected Papers}, edited by Masahito
  Hayashi (World Scientific, Singapore, 2005) Chap. 8,
  pp.~100–112}.}\BibitemShut {Stop}%
\bibitem [{\citenamefont {Hayashi}(2024)}]{hayashi24}%
  \BibitemOpen
  \bibfield  {author} {\bibinfo {author} {\bibfnamefont {Masahito}\
  \bibnamefont {Hayashi}},\ }\bibfield  {title} {\enquote {\bibinfo {title}
  {Alexander {{S}}. {{Holevo}}'s researches in quantum information theory in
  20th century},}\ }\href {\doibase 10.1142/S0219749924400069} {\bibfield
  {journal} {\bibinfo  {journal} {International Journal of Quantum
  Information}\ }\textbf {\bibinfo {volume} {22}},\ \bibinfo {pages} {2440006}
  (\bibinfo {year} {2024})}\BibitemShut {NoStop}%
\bibitem [{\citenamefont {Kahn}\ and\ \citenamefont {Guţă}(2009)}]{kahn09}%
  \BibitemOpen
  \bibfield  {author} {\bibinfo {author} {\bibfnamefont {Jonas}\ \bibnamefont
  {Kahn}}\ and\ \bibinfo {author} {\bibfnamefont {Mădălin}\ \bibnamefont
  {Guţă}},\ }\bibfield  {title} {\enquote {\bibinfo {title} {Local
  {{Asymptotic Normality}} for {{Finite Dimensional Quantum Systems}}},}\
  }\href {\doibase 10.1007/s00220-009-0787-3} {\bibfield  {journal} {\bibinfo
  {journal} {Communications in Mathematical Physics}\ }\textbf {\bibinfo
  {volume} {289}},\ \bibinfo {pages} {597–652} (\bibinfo {year}
  {2009})}\BibitemShut {NoStop}%
\bibitem [{\citenamefont {Gill}\ and\ \citenamefont
  {Guţă}(2013)}]{gill_guta}%
  \BibitemOpen
  \bibfield  {author} {\bibinfo {author} {\bibfnamefont {Richard~D.}\
  \bibnamefont {Gill}}\ and\ \bibinfo {author} {\bibfnamefont {Mădălin}\
  \bibnamefont {Guţă}},\ }\bibfield  {title} {\enquote {\bibinfo {title} {On
  asymptotic quantum statistical inference},}\ }in\ \href {\doibase
  10.1214/12-IMSCOLL909} {\emph {\bibinfo {booktitle} {From Probability to
  Statistics and Back: {{High-dimensional}} Models and Processes – a
  Festschrift in Honor of Jon a. {{Wellner}}}}},\ \bibinfo {series}
  {Collections}, Vol.~\bibinfo {volume} {9},\ \bibinfo {editor} {edited by\
  \bibinfo {editor} {\bibfnamefont {M.}~\bibnamefont {Banerjee}}, \bibinfo
  {editor} {\bibfnamefont {F.}~\bibnamefont {Bunea}}, \bibinfo {editor}
  {\bibfnamefont {J.}~\bibnamefont {Huang}}, \bibinfo {editor} {\bibfnamefont
  {V.}~\bibnamefont {Koltchinskii}}, \ and\ \bibinfo {editor} {\bibfnamefont
  {M.~H.}\ \bibnamefont {Maathuis}}}\ (\bibinfo  {publisher} {Institute of
  Mathematical Statistics},\ \bibinfo {address} {Beachwood, Ohio, USA},\
  \bibinfo {year} {2013})\ p.\ \bibinfo {pages} {105–127}\BibitemShut
  {NoStop}%
\bibitem [{\citenamefont {Yamagata}\ \emph {et~al.}(2013)\citenamefont
  {Yamagata}, \citenamefont {Fujiwara},\ and\ \citenamefont
  {Gill}}]{yamagata13}%
  \BibitemOpen
  \bibfield  {author} {\bibinfo {author} {\bibfnamefont {Koichi}\ \bibnamefont
  {Yamagata}}, \bibinfo {author} {\bibfnamefont {Akio}\ \bibnamefont
  {Fujiwara}}, \ and\ \bibinfo {author} {\bibfnamefont {Richard~D.}\
  \bibnamefont {Gill}},\ }\bibfield  {title} {\enquote {\bibinfo {title}
  {Quantum local asymptotic normality based on a new quantum likelihood
  ratio},}\ }\href {\doibase 10.1214/13-AOS1147} {\bibfield  {journal}
  {\bibinfo  {journal} {The Annals of Statistics}\ }\textbf {\bibinfo {volume}
  {41}},\ \bibinfo {pages} {2197–2217} (\bibinfo {year} {2013})}\BibitemShut
  {NoStop}%
\bibitem [{\citenamefont {Yang}\ \emph {et~al.}(2019)\citenamefont {Yang},
  \citenamefont {Chiribella},\ and\ \citenamefont {Hayashi}}]{yang19}%
  \BibitemOpen
  \bibfield  {author} {\bibinfo {author} {\bibfnamefont {Yuxiang}\ \bibnamefont
  {Yang}}, \bibinfo {author} {\bibfnamefont {Giulio}\ \bibnamefont
  {Chiribella}}, \ and\ \bibinfo {author} {\bibfnamefont {Masahito}\
  \bibnamefont {Hayashi}},\ }\bibfield  {title} {\enquote {\bibinfo {title}
  {Attaining the {{Ultimate Precision Limit}} in {{Quantum State
  Estimation}}},}\ }\href {\doibase 10.1007/s00220-019-03433-4} {\bibfield
  {journal} {\bibinfo  {journal} {Communications in Mathematical Physics}\
  }\textbf {\bibinfo {volume} {368}},\ \bibinfo {pages} {223–293} (\bibinfo
  {year} {2019})}\BibitemShut {NoStop}%
\bibitem [{\citenamefont {Ohya}\ and\ \citenamefont {Petz}(1993)}]{ohya}%
  \BibitemOpen
  \bibfield  {author} {\bibinfo {author} {\bibfnamefont {M.}~\bibnamefont
  {Ohya}}\ and\ \bibinfo {author} {\bibfnamefont {Denes}\ \bibnamefont
  {Petz}},\ }\href {http://www.springer.com/gp/book/9783540208068} {\emph
  {\bibinfo {title} {Quantum {{Entropy}} and {{Its Use}}}}}\ (\bibinfo
  {publisher} {Springer-Verlag},\ \bibinfo {address} {Berlin Heidelberg},\
  \bibinfo {year} {1993})\BibitemShut {NoStop}%
\bibitem [{\citenamefont {Verbeure}(2011)}]{verbeure}%
  \BibitemOpen
  \bibfield  {author} {\bibinfo {author} {\bibfnamefont {Andréé~F.}\
  \bibnamefont {Verbeure}},\ }\href {\doibase 10.1007/978-0-85729-109-7} {\emph
  {\bibinfo {title} {Many-Body Boson Systems}}}\ (\bibinfo  {publisher}
  {Springer},\ \bibinfo {address} {London, England, UK},\ \bibinfo {year}
  {2011})\BibitemShut {NoStop}%
\bibitem [{\citenamefont {Zhou}\ and\ \citenamefont {Chen}(2026)}]{zhou26}%
  \BibitemOpen
  \bibfield  {author} {\bibinfo {author} {\bibfnamefont {Sisi}\ \bibnamefont
  {Zhou}}\ and\ \bibinfo {author} {\bibfnamefont {Senrui}\ \bibnamefont
  {Chen}},\ }\bibfield  {title} {\enquote {\bibinfo {title} {Randomized
  measurements for multiparameter quantum metrology},}\ }\href {\doibase
  10.1103/s27y-gbrp} {\bibfield  {journal} {\bibinfo  {journal} {PRX Quantum}\
  }\textbf {\bibinfo {volume} {7}},\ \bibinfo {pages} {010314} (\bibinfo {year}
  {2026})}\BibitemShut {NoStop}%
\bibitem [{\citenamefont {Matsumoto}(2002)}]{matsumoto02}%
  \BibitemOpen
  \bibfield  {author} {\bibinfo {author} {\bibfnamefont {K.}~\bibnamefont
  {Matsumoto}},\ }\bibfield  {title} {\enquote {\bibinfo {title} {A new
  approach to the {{Cramér--Rao-type}} bound of the pure-state model},}\ }\href
  {\doibase 10.1088/0305-4470/35/13/307} {\bibfield  {journal} {\bibinfo
  {journal} {Journal of Physics A: Mathematical and General}\ }\textbf
  {\bibinfo {volume} {35}},\ \bibinfo {pages} {3111–3123} (\bibinfo {year}
  {2002})}\BibitemShut {NoStop}%
\bibitem [{\citenamefont {Tsang}\ \emph {et~al.}(2011)\citenamefont {Tsang},
  \citenamefont {Wiseman},\ and\ \citenamefont {Caves}}]{twc}%
  \BibitemOpen
  \bibfield  {author} {\bibinfo {author} {\bibfnamefont {Mankei}\ \bibnamefont
  {Tsang}}, \bibinfo {author} {\bibfnamefont {Howard~M.}\ \bibnamefont
  {Wiseman}}, \ and\ \bibinfo {author} {\bibfnamefont {Carlton~M.}\
  \bibnamefont {Caves}},\ }\bibfield  {title} {\enquote {\bibinfo {title}
  {Fundamental {{Quantum Limit}} to {{Waveform Estimation}}},}\ }\href
  {\doibase 10.1103/PhysRevLett.106.090401} {\bibfield  {journal} {\bibinfo
  {journal} {Physical Review Letters}\ }\textbf {\bibinfo {volume} {106}},\
  \bibinfo {pages} {090401} (\bibinfo {year} {2011})}\BibitemShut {NoStop}%
\bibitem [{\citenamefont {Vidrighin}\ \emph {et~al.}(2014)\citenamefont
  {Vidrighin}, \citenamefont {Donati}, \citenamefont {Genoni}, \citenamefont
  {Jin}, \citenamefont {Kolthammer}, \citenamefont {Kim}, \citenamefont
  {Datta}, \citenamefont {Barbieri},\ and\ \citenamefont
  {Walmsley}}]{vidrighin}%
  \BibitemOpen
  \bibfield  {author} {\bibinfo {author} {\bibfnamefont {Mihai~D.}\
  \bibnamefont {Vidrighin}}, \bibinfo {author} {\bibfnamefont {Gaia}\
  \bibnamefont {Donati}}, \bibinfo {author} {\bibfnamefont {Marco~G.}\
  \bibnamefont {Genoni}}, \bibinfo {author} {\bibfnamefont {Xian-Min}\
  \bibnamefont {Jin}}, \bibinfo {author} {\bibfnamefont {W.~Steven}\
  \bibnamefont {Kolthammer}}, \bibinfo {author} {\bibfnamefont {M.~S.}\
  \bibnamefont {Kim}}, \bibinfo {author} {\bibfnamefont {Animesh}\ \bibnamefont
  {Datta}}, \bibinfo {author} {\bibfnamefont {Marco}\ \bibnamefont {Barbieri}},
  \ and\ \bibinfo {author} {\bibfnamefont {Ian~A.}\ \bibnamefont {Walmsley}},\
  }\bibfield  {title} {\enquote {\bibinfo {title} {Joint estimation of phase
  and phase diffusion for quantum metrology},}\ }\href {\doibase
  10.1038/ncomms4532} {\bibfield  {journal} {\bibinfo  {journal} {Nature
  Communications}\ }\textbf {\bibinfo {volume} {5}},\ \bibinfo {pages} {3532}
  (\bibinfo {year} {2014})}\BibitemShut {NoStop}%
\bibitem [{\citenamefont {Genoni}\ \emph {et~al.}(2016)\citenamefont {Genoni},
  \citenamefont {Lami},\ and\ \citenamefont {Serafini}}]{genoni16}%
  \BibitemOpen
  \bibfield  {author} {\bibinfo {author} {\bibfnamefont {Marco~G.}\
  \bibnamefont {Genoni}}, \bibinfo {author} {\bibfnamefont {Ludovico}\
  \bibnamefont {Lami}}, \ and\ \bibinfo {author} {\bibfnamefont {Alessio}\
  \bibnamefont {Serafini}},\ }\bibfield  {title} {\enquote {\bibinfo {title}
  {Conditional and unconditional {{Gaussian}} quantum dynamics},}\ }\href
  {\doibase 10.1080/00107514.2015.1125624} {\bibfield  {journal} {\bibinfo
  {journal} {Contemporary Physics}\ }\textbf {\bibinfo {volume} {57}},\
  \bibinfo {pages} {331–349} (\bibinfo {year} {2016})}\BibitemShut {NoStop}%
\bibitem [{\citenamefont {Tsang}\ \emph {et~al.}(2016)\citenamefont {Tsang},
  \citenamefont {Nair},\ and\ \citenamefont {Lu}}]{tnl}%
  \BibitemOpen
  \bibfield  {author} {\bibinfo {author} {\bibfnamefont {Mankei}\ \bibnamefont
  {Tsang}}, \bibinfo {author} {\bibfnamefont {Ranjith}\ \bibnamefont {Nair}}, \
  and\ \bibinfo {author} {\bibfnamefont {Xiao-Ming}\ \bibnamefont {Lu}},\
  }\bibfield  {title} {\enquote {\bibinfo {title} {Quantum theory of
  superresolution for two incoherent optical point sources},}\ }\href {\doibase
  10.1103/PhysRevX.6.031033} {\bibfield  {journal} {\bibinfo  {journal}
  {Physical Review X}\ }\textbf {\bibinfo {volume} {6}},\ \bibinfo {pages}
  {031033} (\bibinfo {year} {2016})}\BibitemShut {NoStop}%
\bibitem [{\citenamefont {Parniak}\ \emph {et~al.}(2018)\citenamefont
  {Parniak}, \citenamefont {Boroszko}, \citenamefont {Wasilewski},
  \citenamefont {Banaszek},\ and\ \citenamefont
  {{Demkowicz-Dobrzański}}}]{parniak18}%
  \BibitemOpen
  \bibfield  {author} {\bibinfo {author} {\bibfnamefont {Sebastian}\
  \bibnamefont {Parniak}, \bibfnamefont {Michałand~Borówka}}, \bibinfo
  {author} {\bibfnamefont {Kajetan}\ \bibnamefont {Boroszko}}, \bibinfo
  {author} {\bibfnamefont {Wojciech}\ \bibnamefont {Wasilewski}}, \bibinfo
  {author} {\bibfnamefont {Konrad}\ \bibnamefont {Banaszek}}, \ and\ \bibinfo
  {author} {\bibfnamefont {Rafał}\ \bibnamefont {{Demkowicz-Dobrzański}}},\
  }\bibfield  {title} {\enquote {\bibinfo {title} {Beating the {{Rayleigh Limit
  Using Two-Photon Interference}}},}\ }\href {\doibase
  10.1103/PhysRevLett.121.250503} {\bibfield  {journal} {\bibinfo  {journal}
  {Physical Review Letters}\ }\textbf {\bibinfo {volume} {121}},\ \bibinfo
  {pages} {250503} (\bibinfo {year} {2018})}\BibitemShut {NoStop}%
\bibitem [{\citenamefont {Hou}\ \emph {et~al.}(2018)\citenamefont {Hou},
  \citenamefont {Tang}, \citenamefont {Shang}, \citenamefont {Zhu},
  \citenamefont {Li}, \citenamefont {Yuan}, \citenamefont {Wu}, \citenamefont
  {Xiang}, \citenamefont {Li},\ and\ \citenamefont {Guo}}]{hou18}%
  \BibitemOpen
  \bibfield  {author} {\bibinfo {author} {\bibfnamefont {Zhibo}\ \bibnamefont
  {Hou}}, \bibinfo {author} {\bibfnamefont {Jun-Feng}\ \bibnamefont {Tang}},
  \bibinfo {author} {\bibfnamefont {Jiangwei}\ \bibnamefont {Shang}}, \bibinfo
  {author} {\bibfnamefont {Huangjun}\ \bibnamefont {Zhu}}, \bibinfo {author}
  {\bibfnamefont {Jian}\ \bibnamefont {Li}}, \bibinfo {author} {\bibfnamefont
  {Yuan}\ \bibnamefont {Yuan}}, \bibinfo {author} {\bibfnamefont {Kang-Da}\
  \bibnamefont {Wu}}, \bibinfo {author} {\bibfnamefont {Guo-Yong}\ \bibnamefont
  {Xiang}}, \bibinfo {author} {\bibfnamefont {Chuan-Feng}\ \bibnamefont {Li}},
  \ and\ \bibinfo {author} {\bibfnamefont {Guang-Can}\ \bibnamefont {Guo}},\
  }\bibfield  {title} {\enquote {\bibinfo {title} {Deterministic realization of
  collective measurements via photonic quantum walks},}\ }\href {\doibase
  10.1038/s41467-018-03849-x} {\bibfield  {journal} {\bibinfo  {journal}
  {Nature Communications}\ }\textbf {\bibinfo {volume} {9}},\ \bibinfo {pages}
  {1414} (\bibinfo {year} {2018})}\BibitemShut {NoStop}%
\bibitem [{\citenamefont {Albarelli}\ \emph {et~al.}(2020)\citenamefont
  {Albarelli}, \citenamefont {Barbieri}, \citenamefont {Genoni},\ and\
  \citenamefont {Gianani}}]{albarelli20a}%
  \BibitemOpen
  \bibfield  {author} {\bibinfo {author} {\bibfnamefont {F.}~\bibnamefont
  {Albarelli}}, \bibinfo {author} {\bibfnamefont {M.}~\bibnamefont {Barbieri}},
  \bibinfo {author} {\bibfnamefont {M.~G.}\ \bibnamefont {Genoni}}, \ and\
  \bibinfo {author} {\bibfnamefont {I.}~\bibnamefont {Gianani}},\ }\bibfield
  {title} {\enquote {\bibinfo {title} {A perspective on multiparameter quantum
  metrology: {{From}} theoretical tools to applications in quantum imaging},}\
  }\href {\doibase 10.1016/j.physleta.2020.126311} {\bibfield  {journal}
  {\bibinfo  {journal} {Physics Letters A}\ }\textbf {\bibinfo {volume}
  {384}},\ \bibinfo {pages} {126311} (\bibinfo {year} {2020})}\BibitemShut
  {NoStop}%
\bibitem [{\citenamefont {Ansari}\ \emph {et~al.}(2021)\citenamefont {Ansari},
  \citenamefont {Brecht}, \citenamefont {{Gil-Lopez}}, \citenamefont {Donohue},
  \citenamefont {Řeháček}, \citenamefont {Hradil}, \citenamefont
  {{Sánchez-Soto}},\ and\ \citenamefont {Silberhorn}}]{ansari21}%
  \BibitemOpen
  \bibfield  {author} {\bibinfo {author} {\bibfnamefont {Vahid}\ \bibnamefont
  {Ansari}}, \bibinfo {author} {\bibfnamefont {Benjamin}\ \bibnamefont
  {Brecht}}, \bibinfo {author} {\bibfnamefont {Jano}\ \bibnamefont
  {{Gil-Lopez}}}, \bibinfo {author} {\bibfnamefont {John~M.}\ \bibnamefont
  {Donohue}}, \bibinfo {author} {\bibfnamefont {Jaroslav}\ \bibnamefont
  {Řeháček}}, \bibinfo {author} {\bibfnamefont {Zdeněk}\ \bibnamefont
  {Hradil}}, \bibinfo {author} {\bibfnamefont {Luis~L.}\ \bibnamefont
  {{Sánchez-Soto}}}, \ and\ \bibinfo {author} {\bibfnamefont {Christine}\
  \bibnamefont {Silberhorn}},\ }\bibfield  {title} {\enquote {\bibinfo {title}
  {Achieving the {{Ultimate Quantum Timing Resolution}}},}\ }\href {\doibase
  10.1103/PRXQuantum.2.010301} {\bibfield  {journal} {\bibinfo  {journal} {PRX
  Quantum}\ }\textbf {\bibinfo {volume} {2}},\ \bibinfo {pages} {010301}
  (\bibinfo {year} {2021})}\BibitemShut {NoStop}%
\bibitem [{\citenamefont {Chen}\ \emph {et~al.}(2022)\citenamefont {Chen},
  \citenamefont {Chen},\ and\ \citenamefont {Yuan}}]{chen22}%
  \BibitemOpen
  \bibfield  {author} {\bibinfo {author} {\bibfnamefont {Hongzhen}\
  \bibnamefont {Chen}}, \bibinfo {author} {\bibfnamefont {Yu}~\bibnamefont
  {Chen}}, \ and\ \bibinfo {author} {\bibfnamefont {Haidong}\ \bibnamefont
  {Yuan}},\ }\bibfield  {title} {\enquote {\bibinfo {title} {Information
  geometry under hierarchical quantum measurement},}\ }\href {\doibase
  10.1103/PhysRevLett.128.250502} {\bibfield  {journal} {\bibinfo  {journal}
  {Physical Review Letters}\ }\textbf {\bibinfo {volume} {128}},\ \bibinfo
  {pages} {250502} (\bibinfo {year} {2022})}\BibitemShut {NoStop}%
\bibitem [{\citenamefont {Chen}\ \emph {et~al.}(2024)\citenamefont {Chen},
  \citenamefont {Wang},\ and\ \citenamefont {Yuan}}]{chen24}%
  \BibitemOpen
  \bibfield  {author} {\bibinfo {author} {\bibfnamefont {Hongzhen}\
  \bibnamefont {Chen}}, \bibinfo {author} {\bibfnamefont {Lingna}\ \bibnamefont
  {Wang}}, \ and\ \bibinfo {author} {\bibfnamefont {Haidong}\ \bibnamefont
  {Yuan}},\ }\bibfield  {title} {\enquote {\bibinfo {title} {Simultaneous
  measurement of multiple incompatible observables and tradeoff in
  multiparameter quantum estimation},}\ }\href {\doibase
  10.1038/s41534-024-00894-x} {\bibfield  {journal} {\bibinfo  {journal} {npj
  Quantum Information}\ }\textbf {\bibinfo {volume} {10}},\ \bibinfo {pages}
  {98} (\bibinfo {year} {2024})}\BibitemShut {NoStop}%
\bibitem [{\citenamefont {Pelecanos}\ \emph {et~al.}(2025)\citenamefont
  {Pelecanos}, \citenamefont {Spilecki},\ and\ \citenamefont
  {Wright}}]{pelecanos25}%
  \BibitemOpen
  \bibfield  {author} {\bibinfo {author} {\bibfnamefont {Angelos}\ \bibnamefont
  {Pelecanos}}, \bibinfo {author} {\bibfnamefont {Jack}\ \bibnamefont
  {Spilecki}}, \ and\ \bibinfo {author} {\bibfnamefont {John}\ \bibnamefont
  {Wright}},\ }\bibfield  {title} {\enquote {\bibinfo {title} {The debiased
  {{Keyl}}'s algorithm: A new unbiased estimator for full state tomography},}\
  }\href {\doibase 10.48550/arXiv.2510.07788} {\bibfield  {journal} {\bibinfo
  {journal} {ArXiv e-prints}\ } (\bibinfo {year} {2025}),\
  10.48550/arXiv.2510.07788},\ \Eprint {http://arxiv.org/abs/2510.07788}
  {2510.07788} \BibitemShut {NoStop}%
\bibitem [{\citenamefont {Tsang}\ \emph {et~al.}(2020)\citenamefont {Tsang},
  \citenamefont {Albarelli},\ and\ \citenamefont {Datta}}]{semi_prx}%
  \BibitemOpen
  \bibfield  {author} {\bibinfo {author} {\bibfnamefont {Mankei}\ \bibnamefont
  {Tsang}}, \bibinfo {author} {\bibfnamefont {Francesco}\ \bibnamefont
  {Albarelli}}, \ and\ \bibinfo {author} {\bibfnamefont {Animesh}\ \bibnamefont
  {Datta}},\ }\bibfield  {title} {\enquote {\bibinfo {title} {Quantum
  {{Semiparametric Estimation}}},}\ }\href {\doibase
  10.1103/PhysRevX.10.031023} {\bibfield  {journal} {\bibinfo  {journal}
  {Physical Review X}\ }\textbf {\bibinfo {volume} {10}},\ \bibinfo {pages}
  {031023} (\bibinfo {year} {2020})}\BibitemShut {NoStop}%
\bibitem [{\citenamefont {Tsang}(2021)}]{qlmoment_pra2}%
  \BibitemOpen
  \bibfield  {author} {\bibinfo {author} {\bibfnamefont {Mankei}\ \bibnamefont
  {Tsang}},\ }\bibfield  {title} {\enquote {\bibinfo {title} {Quantum limit to
  subdiffraction incoherent optical imaging. {{II}}. {{A}} parametric-submodel
  approach},}\ }\href {\doibase 10.1103/PhysRevA.104.052411} {\bibfield
  {journal} {\bibinfo  {journal} {Physical Review A}\ }\textbf {\bibinfo
  {volume} {104}},\ \bibinfo {pages} {052411} (\bibinfo {year}
  {2021})}\BibitemShut {NoStop}%
\bibitem [{\citenamefont {Fujiwara}(2005)}]{fujiwara05}%
  \BibitemOpen
  \bibfield  {author} {\bibinfo {author} {\bibfnamefont {Akio}\ \bibnamefont
  {Fujiwara}},\ }\bibfield  {title} {\enquote {\bibinfo {title} {Geometry of
  quantum estimation theory},}\ }in\ \href {\doibase
  10.1142/9789812563071_0019} {\emph {\bibinfo {booktitle} {Asymptotic Theory
  of Quantum Statistical Inference: {{Selected}} Papers}}},\ \bibinfo {editor}
  {edited by\ \bibinfo {editor} {\bibfnamefont {Masahito}\ \bibnamefont
  {Hayashi}}}\ (\bibinfo  {publisher} {World Scientific},\ \bibinfo {address}
  {Singapore},\ \bibinfo {year} {2005})\ Chap.~\bibinfo {chapter} {18}, p.\
  \bibinfo {pages} {229–286}\BibitemShut {NoStop}%
\bibitem [{\citenamefont {Carollo}\ \emph {et~al.}(2019)\citenamefont
  {Carollo}, \citenamefont {Spagnolo}, \citenamefont {Dubkov},\ and\
  \citenamefont {Valenti}}]{carollo19}%
  \BibitemOpen
  \bibfield  {author} {\bibinfo {author} {\bibfnamefont {Angelo}\ \bibnamefont
  {Carollo}}, \bibinfo {author} {\bibfnamefont {Bernardo}\ \bibnamefont
  {Spagnolo}}, \bibinfo {author} {\bibfnamefont {Alexander~A.}\ \bibnamefont
  {Dubkov}}, \ and\ \bibinfo {author} {\bibfnamefont {Davide}\ \bibnamefont
  {Valenti}},\ }\bibfield  {title} {\enquote {\bibinfo {title} {On quantumness
  in multi-parameter quantum estimation},}\ }\href {\doibase
  10.1088/1742-5468/ab3ccb} {\bibfield  {journal} {\bibinfo  {journal} {Journal
  of Statistical Mechanics: Theory and Experiment}\ }\textbf {\bibinfo {volume}
  {2019}},\ \bibinfo {pages} {094010} (\bibinfo {year} {2019})}\BibitemShut
  {NoStop}%
\bibitem [{\citenamefont {Carollo}\ \emph {et~al.}(2020)\citenamefont
  {Carollo}, \citenamefont {Spagnolo}, \citenamefont {Dubkov},\ and\
  \citenamefont {Valenti}}]{carollo20}%
  \BibitemOpen
  \bibfield  {author} {\bibinfo {author} {\bibfnamefont {Angelo}\ \bibnamefont
  {Carollo}}, \bibinfo {author} {\bibfnamefont {Bernardo}\ \bibnamefont
  {Spagnolo}}, \bibinfo {author} {\bibfnamefont {Alexander~A.}\ \bibnamefont
  {Dubkov}}, \ and\ \bibinfo {author} {\bibfnamefont {Davide}\ \bibnamefont
  {Valenti}},\ }\bibfield  {title} {\enquote {\bibinfo {title} {Erratum: {{On}}
  quantumness in multi-parameter quantum estimation (2019 {{J}}. {{Stat}}.
  {{Mech}}. 094010)},}\ }\href {\doibase 10.1088/1742-5468/ab6f5e} {\bibfield
  {journal} {\bibinfo  {journal} {Journal of Statistical Mechanics: Theory and
  Experiment}\ }\textbf {\bibinfo {volume} {2020}},\ \bibinfo {pages} {029902}
  (\bibinfo {year} {2020})}\BibitemShut {NoStop}%
\bibitem [{\citenamefont {{Van Trees}}(2003)}]{vantrees2}%
  \BibitemOpen
  \bibfield  {author} {\bibinfo {author} {\bibfnamefont {Harry~L.}\
  \bibnamefont {{Van Trees}}},\ }\href
  {http://books.google.com/books?id=4HtGAAAAYAAJ} {\emph {\bibinfo {title}
  {Detection, Estimation, and Modulation Theory, Part {{II}}: {{Nonlinear}}
  Modulation Theory}}}\ (\bibinfo  {publisher} {John Wiley \& Sons},\ \bibinfo
  {address} {New York},\ \bibinfo {year} {2003})\BibitemShut {NoStop}%
\bibitem [{\citenamefont {Wiseman}\ and\ \citenamefont
  {Milburn}(2010)}]{wiseman_milburn}%
  \BibitemOpen
  \bibfield  {author} {\bibinfo {author} {\bibfnamefont {Howard~M.}\
  \bibnamefont {Wiseman}}\ and\ \bibinfo {author} {\bibfnamefont {Gerard~J.}\
  \bibnamefont {Milburn}},\ }\href {\doibase 10.1017/CBO9780511813948} {\emph
  {\bibinfo {title} {Quantum Measurement and Control}}}\ (\bibinfo  {publisher}
  {Cambridge University Press},\ \bibinfo {address} {Cambridge},\ \bibinfo
  {year} {2010})\BibitemShut {NoStop}%
\bibitem [{\citenamefont {Larson}\ and\ \citenamefont
  {Mavrogordatos}(2024)}]{larson}%
  \BibitemOpen
  \bibfield  {author} {\bibinfo {author} {\bibfnamefont {Jonas}\ \bibnamefont
  {Larson}}\ and\ \bibinfo {author} {\bibfnamefont {Themistoklis}\ \bibnamefont
  {Mavrogordatos}},\ }\href {\doibase 10.1088/978-0-7503-6452-2} {\emph
  {\bibinfo {title} {The Jaynes–Cummings Model and Its Descendants}}},\
  \bibinfo {edition} {2nd}\ ed.\ (\bibinfo  {publisher} {IOP Publishing},\
  \bibinfo {address} {Bristol, England, UK},\ \bibinfo {year}
  {2024})\BibitemShut {NoStop}%
\bibitem [{\citenamefont {Miyazaki}\ and\ \citenamefont
  {Matsumoto}(2022)}]{miyazaki22}%
  \BibitemOpen
  \bibfield  {author} {\bibinfo {author} {\bibfnamefont {Jisho}\ \bibnamefont
  {Miyazaki}}\ and\ \bibinfo {author} {\bibfnamefont {Keiji}\ \bibnamefont
  {Matsumoto}},\ }\bibfield  {title} {\enquote {\bibinfo {title}
  {Imaginarity-free quantum multiparameter estimation},}\ }\href {\doibase
  10.22331/q-2022-03-10-665} {\bibfield  {journal} {\bibinfo  {journal}
  {Quantum}\ }\textbf {\bibinfo {volume} {6}},\ \bibinfo {pages} {665}
  (\bibinfo {year} {2022})},\ \Eprint {http://arxiv.org/abs/2010.15465v3}
  {2010.15465v3} \BibitemShut {NoStop}%
\bibitem [{\citenamefont {Holevo}(1977)}]{holevo77}%
  \BibitemOpen
  \bibfield  {author} {\bibinfo {author} {\bibfnamefont {A.~S.}\ \bibnamefont
  {Holevo}},\ }\bibfield  {title} {\enquote {\bibinfo {title} {Commutation
  superoperator of a state and its applications to the noncommutative
  statistics},}\ }\href {\doibase 10.1016/0034-4877(77)90009-X} {\bibfield
  {journal} {\bibinfo  {journal} {Reports on Mathematical Physics}\ }\textbf
  {\bibinfo {volume} {12}},\ \bibinfo {pages} {251–271} (\bibinfo {year}
  {1977})}\BibitemShut {NoStop}%
\bibitem [{\citenamefont {Horn}\ and\ \citenamefont {Johnson}(2013)}]{horn}%
  \BibitemOpen
  \bibfield  {author} {\bibinfo {author} {\bibfnamefont {Roger~A.}\
  \bibnamefont {Horn}}\ and\ \bibinfo {author} {\bibfnamefont {Charles~R.}\
  \bibnamefont {Johnson}},\ }\href@noop {} {\emph {\bibinfo {title} {Matrix
  Analysis}}},\ \bibinfo {edition} {2nd}\ ed.\ (\bibinfo  {publisher}
  {Cambridge University Press},\ \bibinfo {address} {Cambridge},\ \bibinfo
  {year} {2013})\BibitemShut {NoStop}%
\bibitem [{sm()}]{sm}%
  \BibitemOpen
  \href@noop {} {}\bibinfo {note} {The Supplemental Material proves that, to
  achieve the Holevo–Nagaoka bound, the required state of the general-dyne
  ancillas is always physical, and if $W = I$, the state is the vacuum state,
  meaning that the general-dyne measurement is simply heterodyne.}\BibitemShut
  {Stop}%
\bibitem [{\citenamefont {Braginsky}\ and\ \citenamefont
  {Khalili}(1992)}]{braginsky}%
  \BibitemOpen
  \bibfield  {author} {\bibinfo {author} {\bibfnamefont {Vladimir~B.}\
  \bibnamefont {Braginsky}}\ and\ \bibinfo {author} {\bibfnamefont {Farid~Ya.}\
  \bibnamefont {Khalili}},\ }\href@noop {} {\emph {\bibinfo {title} {Quantum
  Measurement}}}\ (\bibinfo  {publisher} {Cambridge University Press},\
  \bibinfo {address} {Cambridge},\ \bibinfo {year} {1992})\BibitemShut
  {NoStop}%
\bibitem [{\citenamefont {Grangier}\ \emph {et~al.}(1998)\citenamefont
  {Grangier}, \citenamefont {Levenson},\ and\ \citenamefont
  {Poizat}}]{grangier98}%
  \BibitemOpen
  \bibfield  {author} {\bibinfo {author} {\bibfnamefont {Philippe}\
  \bibnamefont {Grangier}}, \bibinfo {author} {\bibfnamefont {Juan~Ariel}\
  \bibnamefont {Levenson}}, \ and\ \bibinfo {author} {\bibfnamefont
  {Jean-Philippe}\ \bibnamefont {Poizat}},\ }\bibfield  {title} {\enquote
  {\bibinfo {title} {Quantum non-demolition measurements in optics},}\ }\href
  {\doibase 10.1038/25059} {\bibfield  {journal} {\bibinfo  {journal} {Nature}\
  }\textbf {\bibinfo {volume} {396}},\ \bibinfo {pages} {537–542} (\bibinfo
  {year} {1998})}\BibitemShut {NoStop}%
\bibitem [{\citenamefont {Hammerer}\ \emph {et~al.}(2010)\citenamefont
  {Hammerer}, \citenamefont {Sørensen},\ and\ \citenamefont
  {Polzik}}]{hammerer10}%
  \BibitemOpen
  \bibfield  {author} {\bibinfo {author} {\bibfnamefont {Klemens}\ \bibnamefont
  {Hammerer}}, \bibinfo {author} {\bibfnamefont {Anders~S.}\ \bibnamefont
  {Sørensen}}, \ and\ \bibinfo {author} {\bibfnamefont {Eugene~S.}\
  \bibnamefont {Polzik}},\ }\bibfield  {title} {\enquote {\bibinfo {title}
  {Quantum interface between light and atomic ensembles},}\ }\href {\doibase
  10.1103/RevModPhys.82.1041} {\bibfield  {journal} {\bibinfo  {journal}
  {Reviews of Modern Physics}\ }\textbf {\bibinfo {volume} {82}},\ \bibinfo
  {pages} {1041–1093} (\bibinfo {year} {2010})}\BibitemShut {NoStop}%
\bibitem [{\citenamefont {Drummond}\ and\ \citenamefont
  {Hillery}(2014)}]{drummond}%
  \BibitemOpen
  \bibfield  {author} {\bibinfo {author} {\bibfnamefont {Peter~D.}\
  \bibnamefont {Drummond}}\ and\ \bibinfo {author} {\bibfnamefont {Mark}\
  \bibnamefont {Hillery}},\ }\href {\doibase 10.1017/CBO9780511783616} {\emph
  {\bibinfo {title} {The Quantum Theory of Nonlinear Optics}}}\ (\bibinfo
  {publisher} {Cambridge University Press},\ \bibinfo {address} {Cambridge,
  England, UK},\ \bibinfo {year} {2014})\BibitemShut {NoStop}%
\bibitem [{\citenamefont {Firstenberg}\ \emph {et~al.}(2016)\citenamefont
  {Firstenberg}, \citenamefont {Adams},\ and\ \citenamefont
  {Hofferberth}}]{firstenberg16}%
  \BibitemOpen
  \bibfield  {author} {\bibinfo {author} {\bibfnamefont {O.}~\bibnamefont
  {Firstenberg}}, \bibinfo {author} {\bibfnamefont {C.~S.}\ \bibnamefont
  {Adams}}, \ and\ \bibinfo {author} {\bibfnamefont {S.}~\bibnamefont
  {Hofferberth}},\ }\bibfield  {title} {\enquote {\bibinfo {title} {Nonlinear
  quantum optics mediated by {{Rydberg}} interactions},}\ }\href {\doibase
  10.1088/0953-4075/49/15/152003} {\bibfield  {journal} {\bibinfo  {journal}
  {Journal of Physics B: Atomic, Molecular and Optical Physics}\ }\textbf
  {\bibinfo {volume} {49}},\ \bibinfo {pages} {152003} (\bibinfo {year}
  {2016})}\BibitemShut {NoStop}%
\bibitem [{\citenamefont {Gill}\ and\ \citenamefont
  {Massar}(2000)}]{gill_massar}%
  \BibitemOpen
  \bibfield  {author} {\bibinfo {author} {\bibfnamefont {Richard~D.}\
  \bibnamefont {Gill}}\ and\ \bibinfo {author} {\bibfnamefont {Serge}\
  \bibnamefont {Massar}},\ }\bibfield  {title} {\enquote {\bibinfo {title}
  {State estimation for large ensembles},}\ }\href {\doibase
  10.1103/PhysRevA.61.042312} {\bibfield  {journal} {\bibinfo  {journal}
  {Physical Review A}\ }\textbf {\bibinfo {volume} {61}},\ \bibinfo {pages}
  {042312} (\bibinfo {year} {2000})}\BibitemShut {NoStop}%
\bibitem [{\citenamefont {Conlon}\ \emph {et~al.}(2021)\citenamefont {Conlon},
  \citenamefont {Suzuki}, \citenamefont {Lam},\ and\ \citenamefont
  {Assad}}]{conlon21}%
  \BibitemOpen
  \bibfield  {author} {\bibinfo {author} {\bibfnamefont {Lorcán~O.}\
  \bibnamefont {Conlon}}, \bibinfo {author} {\bibfnamefont {Jun}\ \bibnamefont
  {Suzuki}}, \bibinfo {author} {\bibfnamefont {Ping~Koy}\ \bibnamefont {Lam}},
  \ and\ \bibinfo {author} {\bibfnamefont {Syed~M.}\ \bibnamefont {Assad}},\
  }\bibfield  {title} {\enquote {\bibinfo {title} {Efficient computation of the
  {{Nagaoka}}–{{Hayashi}} bound for multiparameter estimation with separable
  measurements},}\ }\href {\doibase 10.1038/s41534-021-00414-1} {\bibfield
  {journal} {\bibinfo  {journal} {npj Quantum Information}\ }\textbf {\bibinfo
  {volume} {7}},\ \bibinfo {pages} {110} (\bibinfo {year} {2021})}\BibitemShut
  {NoStop}%
\bibitem [{\citenamefont {Belliardo}\ and\ \citenamefont
  {Giovannetti}(2021)}]{belliardo21}%
  \BibitemOpen
  \bibfield  {author} {\bibinfo {author} {\bibfnamefont {Federico}\
  \bibnamefont {Belliardo}}\ and\ \bibinfo {author} {\bibfnamefont {Vittorio}\
  \bibnamefont {Giovannetti}},\ }\bibfield  {title} {\enquote {\bibinfo {title}
  {Incompatibility in quantum parameter estimation},}\ }\href {\doibase
  10.1088/1367-2630/ac04ca} {\bibfield  {journal} {\bibinfo  {journal} {New
  Journal of Physics}\ }\textbf {\bibinfo {volume} {23}},\ \bibinfo {pages}
  {063055} (\bibinfo {year} {2021})}\BibitemShut {NoStop}%
\bibitem [{\citenamefont {Goodman}(2004)}]{goodman}%
  \BibitemOpen
  \bibfield  {author} {\bibinfo {author} {\bibfnamefont {Joseph~W.}\
  \bibnamefont {Goodman}},\ }\href@noop {} {\emph {\bibinfo {title}
  {Introduction to Fourier Optics}}}\ (\bibinfo  {publisher} {McGraw-Hill},\
  \bibinfo {address} {New York},\ \bibinfo {year} {2004})\BibitemShut {NoStop}%
\bibitem [{\citenamefont {Holland}\ \emph {et~al.}(1991)\citenamefont
  {Holland}, \citenamefont {Walls},\ and\ \citenamefont {Zoller}}]{holland91}%
  \BibitemOpen
  \bibfield  {author} {\bibinfo {author} {\bibfnamefont {M.~J.}\ \bibnamefont
  {Holland}}, \bibinfo {author} {\bibfnamefont {D.~F.}\ \bibnamefont {Walls}},
  \ and\ \bibinfo {author} {\bibfnamefont {P.}~\bibnamefont {Zoller}},\
  }\bibfield  {title} {\enquote {\bibinfo {title} {Quantum nondemolition
  measurements of photon number by atomic beam deflection},}\ }\href {\doibase
  10.1103/PhysRevLett.67.1716} {\bibfield  {journal} {\bibinfo  {journal}
  {Physical Review Letters}\ }\textbf {\bibinfo {volume} {67}},\ \bibinfo
  {pages} {1716–1719} (\bibinfo {year} {1991})}\BibitemShut {NoStop}%
\bibitem [{\citenamefont {Jacobs}\ \emph {et~al.}(1994)\citenamefont {Jacobs},
  \citenamefont {Tombesi}, \citenamefont {Collett},\ and\ \citenamefont
  {Walls}}]{jacobs94}%
  \BibitemOpen
  \bibfield  {author} {\bibinfo {author} {\bibfnamefont {K.}~\bibnamefont
  {Jacobs}}, \bibinfo {author} {\bibfnamefont {P.}~\bibnamefont {Tombesi}},
  \bibinfo {author} {\bibfnamefont {M.~J.}\ \bibnamefont {Collett}}, \ and\
  \bibinfo {author} {\bibfnamefont {D.~F.}\ \bibnamefont {Walls}},\ }\bibfield
  {title} {\enquote {\bibinfo {title} {Quantum-nondemolition measurement of
  photon number using radiation pressure},}\ }\href {\doibase
  10.1103/PhysRevA.49.1961} {\bibfield  {journal} {\bibinfo  {journal}
  {Physical Review A}\ }\textbf {\bibinfo {volume} {49}},\ \bibinfo {pages}
  {1961–1966} (\bibinfo {year} {1994})}\BibitemShut {NoStop}%
\bibitem [{\citenamefont {Tsang}(2010)}]{cqeo}%
  \BibitemOpen
  \bibfield  {author} {\bibinfo {author} {\bibfnamefont {Mankei}\ \bibnamefont
  {Tsang}},\ }\bibfield  {title} {\enquote {\bibinfo {title} {Cavity quantum
  electro-optics},}\ }\href {\doibase 10.1103/PhysRevA.81.063837} {\bibfield
  {journal} {\bibinfo  {journal} {Physical Review A}\ }\textbf {\bibinfo
  {volume} {81}},\ \bibinfo {pages} {063837} (\bibinfo {year}
  {2010})}\BibitemShut {NoStop}%
\end{thebibliography}

%

\clearpage

\appendix

\section{End Matter}
\subsection{\label{app_imaging}Diffraction-limited incoherent imaging}
With a diffraction-limited imaging system, the point-spread function
for the optical field can be modeled as a real function multiplied by
$\exp(i\eta)$ for some constant global phase $\eta \in \mb R$
\cite{goodman}. Assume a distribution of spatially incoherent sources
on the object plane. Then the mutual coherence matrix $\Upsilon$ with
respect to the wavepacket-mode basis on the image plane is real
\cite[Appendix~B]{tnl}. Assuming that the source distribution depends
on $\theta$, $\Upsilon(\theta)$ is real for all $\theta$.

Let the $m$-spatial-mode coherent state be
\begin{align}
\ket{\alpha} &= \sum_n e^{-\norm{\alpha}^2/2} \frac{\alpha^n}{\sqrt{n!}} \ket{n},
\end{align}
where $\alpha \equiv \mqty(\alpha_1 & \dots & \alpha_m)^\top$,
$n \equiv \mqty(n_1 & \dots & n_m)^\top$,
$\norm{\alpha}^2 \equiv \sum_\mu \abs{\alpha_\mu}^2$,
$\alpha^n \equiv \prod_\mu \alpha_\mu^{n_\mu}$,
$n! \equiv \prod_\mu n_\mu!$, and each $\ket{n}$ is a multimode Fock
state. Let the density operator be
\begin{align}
\rho(\theta) &= \int P(\alpha|\theta) \ket{\alpha}\bra{\alpha}d^{2m}\alpha,
\\
P(\alpha|\theta) &= \frac{1}{\det[\pi\Upsilon(\theta)]} 
\exp\Bk{-\alpha^\dagger\Upsilon(\theta)^{-1}\alpha},
\end{align}
where
$d^{2m}\alpha \equiv \prod_{\mu=1}^m
d(\Re\alpha_\mu)d(\Im\alpha_\mu)$.  Define an antiunitary conjugation
operator $K$ by $K \ket{n} = \ket{n}$.  Then
$K\ket{\alpha} = \ket{\alpha^*}$, where $*$ denotes the entry-wise
complex conjugate, leading to
$K\rho(\theta) K^\dagger = \int P(\alpha|\theta)
\ket{\alpha^*}\bra{\alpha^*} d^{2m}\alpha$.  Since
$\alpha^\dagger\Upsilon^{-1}\alpha =
\bk{\alpha^\dagger\Upsilon^{-1}\alpha }^* =
\alpha^{*\dagger}\Upsilon^{*-1} \alpha^*$ and
$\int f(\alpha) d^{2m}\alpha = \int f(\alpha) d^{2m}\bk{\alpha^*}$, a
change of variables $\beta = \alpha^*$ leads to
$K\rho(\theta) K^\dagger = \rho(\theta)$ if
$\Upsilon(\theta) = \Upsilon(\theta)^*$. The model then satisfies the
so-called global antiunitary symmetry \cite{miyazaki22}. By
\cite[Theorem~2]{miyazaki22},
$\Im \trace[S(\dot\phi) S(\dot\varphi)\rho] = 0$ for any pair of
scores $S(\dot\phi)$ and $S(\dot\varphi)$, and since $\delta^\Hel$ is
a real linear combination of the scores, $\Im \Gamma(\delta^\Hel) = 0$
results. It follows that $C(\delta^\Hel) = \Hel$, and by
Eq.~(\ref{bounds}), $\Hel = \HN = C(\delta^\Hel)$, which implies that
$\delta^\Hel$ is optimal for $\HN$.  These results hold for all
$\theta$.

Following Eq.~(\ref{effinf_thermal}), the optimal
  observables to measure are $X^{(M)}$ with
\begin{align}
X_j = \delta_j^\Hel + \bk{\beta_j-c_j} I = \sum_{\mu,\nu} D_{j\mu\nu} a_\mu^\dagger a_\nu
+ c_j' I
\end{align}
for some Hermitian $D_j \in \mb C^{m\times m}$ and $c_j' \in \mb R$
that can be computed after the acquisition step. A measurement of
$X_j$ can be implemented by measuring $X_j - c_j' I$ and then adding
$c_j'$ to the outcome, so the $c_j'I$ term will be ignored
hereafter. The $M$ IID copies correspond to $M$ temporal optical
modes.  Eqs.~(\ref{ImG})--(\ref{Ycompat}) with
$\Im\Gamma(\delta^\Hel) = 0$ imply that $X^{(M)} = Y^{(M)} = x^{(M)}$,
which are all compatible observables in the $M\to\infty$ limit.  As
the measurements of compatible observables are commonly called QND
\cite{braginsky}, the measurements of asymptotically compatible
observables are hereafter called asymptotically QND (AQND). The
Hamiltonian in Eqs.~(\ref{disp}) becomes
\begin{align}
H^{(M)} &= H_2^{(M)} = \gamma \sum_{j,\mu,\nu} D_{j\mu\nu} \bk{a_\mu^\dagger a_\nu}^{(M)} p_j'',
\end{align}
where $\bk{a_\mu^\dagger a_\nu}^{(M)}$ is in the collective form of
Eq.~(\ref{XM}).  To find a physical implementation of this
Hamiltonian, it is fruitful to diagonalize each $D_j$ as
$D_{j\mu\nu} = \sum_\xi \lambda_{j\xi} U_{j\mu\xi} U_{j\nu\xi}^*$ in
terms of the real eigenvalues $\BK{\lambda_{j\xi}:\xi \in \Omega}$ and
a unitary matrix $U_j$, leading to
\begin{align}
X_j &= \sum_\xi \lambda_{j\xi} a_{j\xi}^\dagger a_{j\xi},
\quad
a_{j\xi} \equiv \sum_\mu U_{j\mu\xi}^* a_\mu,
\\
H^{(M)} &= \gamma \sum_{j,\xi}\lambda_{j\xi} \bk{a_{j\xi}^\dagger a_{j\xi}}^{(M)}  p_j''.
\label{Hb}
\end{align}
It can in fact be shown that both $D_j$ and $U_j$ are real if
$\Upsilon$ is real for all $\theta$.  For a given $X_j$, the $U_j$
matrix defines a new basis of spatial modes such that $X_j$ is a
linear combination of the photon numbers
$\{a_{j\xi}^\dagger a_{j\xi}:\xi \in \Omega\}$ of the modes; I call
these modes the eigenmodes of $X_j$.
$(a_{j\xi}^\dagger a_{j\xi})^{(M)}$ is then the time-integrated photon
number of the $\xi$th eigenmode of $X_j$. To measure $X_j^{(M)}$,
Eq.~(\ref{Hb}) means that the $q_j''$ quadrature of the $j$th ancilla
should be displaced by the photon numbers
$\{(a_{j\xi}^\dagger a_{j\xi})^{(M)}:\xi \in \Omega\}$ of the
eigenmodes of $X_j$ with varying coupling coefficients
$\{\gamma \lambda_{j\xi}:\xi \in \Omega\}$. This type of interaction
has been extensively studied---both theoretically and
experimentally---in the context of QND photon-number measurements
\cite{grangier98}. The interaction may be implemented by the
cross-Kerr effect \cite{grangier98}, dispersive atomic coupling
\cite{holland91}, optomechanics \cite{jacobs94}, or electro-optics
\cite{cqeo}. For example, with the cross-Kerr effect, one can use a
strong coherent optical beam as an ancilla for which $q_j''$ and
$p_j''$ are the linearized phase and intensity quadratures,
respectively. Each ancilla phase quadrature $q_j''$ should then be
modulated by the photon numbers of the eigenmodes of $X_j$ as per
Eq.~(\ref{Hb}). Fig.~\ref{QNDSPADE}(a) sketches the experimental
scheme.

\fig{0.45}{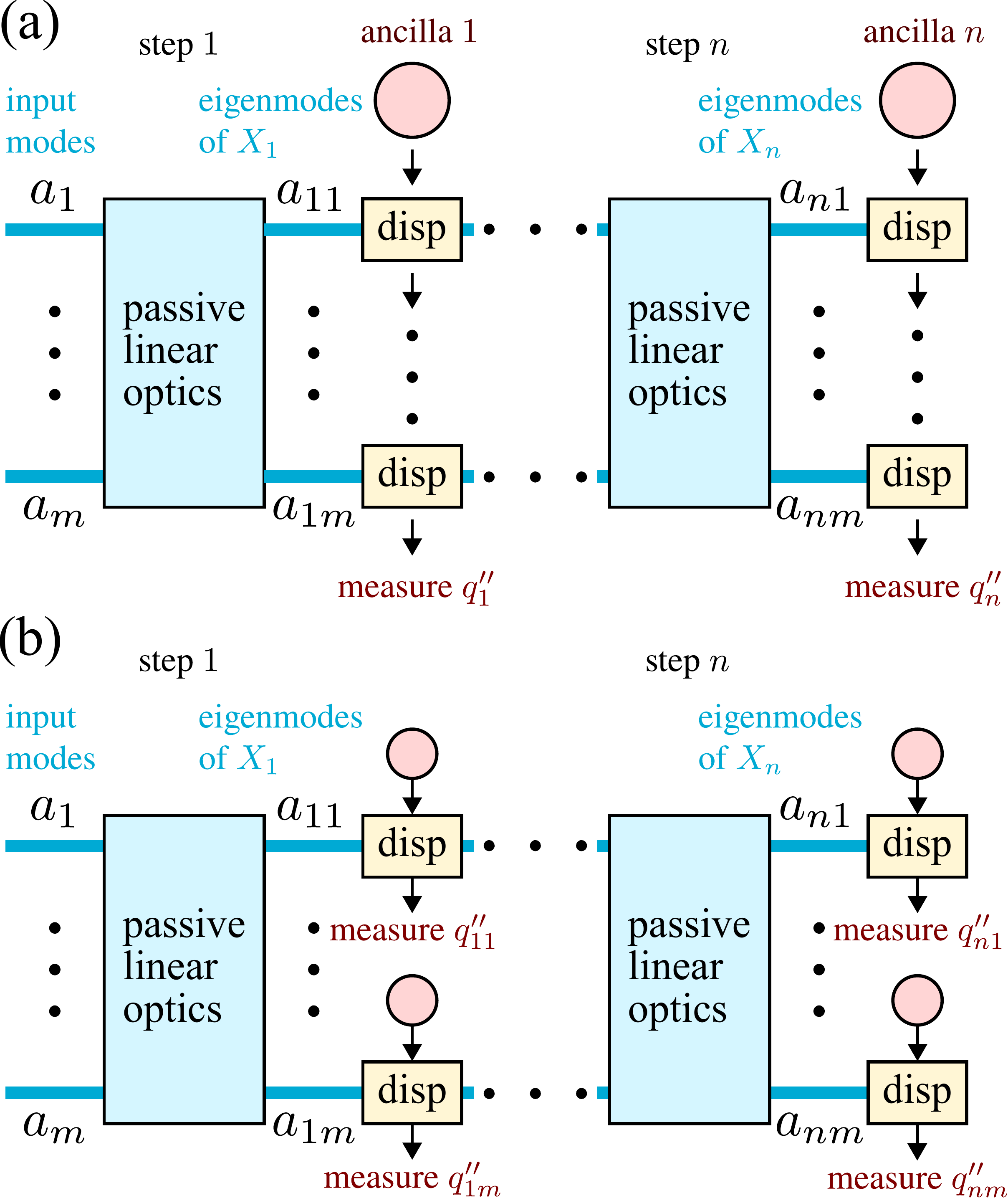}{(a) A scheme to measure $n$ observables
    $X^{(M)}$ for diffraction-limited incoherent imaging.  At the
    $j$th step with $j = 1,\dots,n$, the $\xi$th output channel of the
    passive linear optics corresponds to the $\xi$th eigenmode of
    $X_j$ with annihilation operator $a_{j\xi}$. In each ``disp'' box,
    the $q_j''$ quadrature of the $j$th ancilla is displaced by the
    photon number $(a_{j\xi}^\dagger a_{j\xi})^{(M)}$ of the $\xi$th
    eigenmode of $X_j$. A measurement of $q_j''$ implements an
    indirect measurement of $X_j^{(M)}$. (b) An alternative setup that
    measures each photon number $(a_{j\xi}^\dagger a_{j\xi})^{(M)}$
    with one separate ancilla with quadratures
    $(q_{j\xi}'',p_{j\xi}'')$.}

Since $U_j$ is real for all $j$, it can be shown that the set of all
photon numbers
$\{(a_{j\xi}^\dagger a_{j\xi})^{(M)}:j = 1,\dots,n,\xi \in \Omega\}$
are also asymptotically compatible. An equivalent setup is therefore
to measure each photon number $(a_{j\xi}^\dagger a_{j\xi})^{(M)}$ with
a separate AQND measurement, as illustrated by
Fig.~\ref{QNDSPADE}(b).  Since $X_j^{(M)}$ is a linear combination of
$\{(a_{j\xi}^\dagger a_{j\xi})^{(M)}:\xi \in \Omega\}$, the outcomes
of the $m$ AQND measurements at the $j$th step can be combined to
realize an effective measurement of $X_j^{(M)}$. Although this scheme
requires more ancillas ($nm$ ancilla modes in total), each AQND
measurement involves only one spatial mode and may be easier to
implement. The increased number of outcomes may also give more useful
information about the sources.

It is now clear how the schemes here generalize the previously
proposed spatial-mode-demultiplexing method
\cite{tnl,review_cp,lvovsky26}.  Whereas the previous proposals assume
that the photon numbers of each spatial-mode basis are measured
destructively, the schemes here perform AQND measurements, so that
multiple spatial-mode bases of the same optical field can be measured
with vanishing measurement-backaction noise.

\subsection{\label{app_spin}Nonparametric estimation 
of spin expected values}
Let $s \equiv \mqty(s_1 & s_2 &s_3)^\top$ be spin observables of a
$d$-level object satisfying
$[s_j,s_k] = i\sum_l \varepsilon_{jkl} s_l$ in terms of the
Levi-Civita symbol $\varepsilon$. Note that $d$ may be higher than
$2$. Let $\beta_j(\theta) = \trace[s_j \rho(\theta)]$ be the parameters
of interest. With the nonparametric model $\rho(\theta) = \theta$,
Eq.~(\ref{nonparam}) gives $X = s$.  Suppose that the preliminary
estimate gives $\trace(s_1\check\theta) = r > 0$ and
$\trace(s_2 \check\theta) = \trace(s_3\check\theta) = 0$.  Then
one should set
\begin{align}
Y &= \mqty(q\\ p\\ x) = \frac{1}{\sqrt{r}} \mqty(X_2\\ X_3\\ X_1)
= \frac{1}{\sqrt{r}} \mqty(s_2\\ s_3\\ s_1),
\\
Y^{(M)} &= \mqty(q^{(M)}\\ p^{(M)}\\ x^{(M)}) 
= \frac{1}{\sqrt{M r}} \mqty(J_2^{(M)}\\ J_3^{(M)}\\ J_1^{(M)}),
\end{align}
\begin{align}
J_j^{(M)} &\equiv 
\sum_{k=1}^M I^{\otimes(k-1)}\otimes s_j \otimes I^{\otimes(M-k)}.
\end{align}
For any other preliminary estimate, the relation between $X$ and $Y$
is more complicated but the idea is similar: $(q,p)$ should be set as
the spin components transverse to the estimated spin vector and $x$ as
the longitudinal spin component.  Write
\begin{align}
H^{(M)} &= \kappa
\bk{q^{(M)} p' - p^{(M)} q'} + \gamma x^{(M)} p''.
\end{align}
The exact equations of motion become
\begin{align}
\dv{q^{(M)}(t)}{t} &= -\kappa \frac{J_1^{(M)}(t)}{Mr} q'(t)
-\gamma \frac{J_3^{(M)}(t)}{Mr}  p''(t),
\label{spin1}
\\
\dv{p^{(M)}(t)}{t} &= -\kappa \frac{J_1^{(M)}(t)}{Mr} p'(t)
+\gamma \frac{J_2^{(M)}(t)}{Mr}  p''(t),
\label{spin2}
\\
\dv{x^{(M)}(t)}{t} &= \kappa \frac{J_3^{(M)}(t)}{Mr} p'(t)
+\kappa \frac{J_2^{(M)}(t)}{Mr}  q'(t)
\label{spin3}
\end{align}
for the collective observables and
\begin{align}
\dv{q'(t)}{t} &= \kappa q^{(M)}(t),
&
\dv{p'(t)}{t} &= \kappa p^{(M)}(t),
\\
\dv{q''(t)}{t} &= \gamma x^{(M)}(t),
&
\dv{p''(t)}{t} &= 0
\end{align}
for the ancilla quadratures. Approximating the intensive spin
observables as
$J_j^{(M)}(t)/M = \trace(s_j\check\theta) I + O(\epsilon)$ in terms of
their initial expected values $\trace(s_j\check\theta)$ and some
small $\epsilon > 0$, Eqs.~(\ref{spin1})--(\ref{spin3}) can be
linearized to become
\begin{align}
\dv{q^{(M)}(t)}{t} &= -\kappa q'(t) + O(\epsilon),
\\
\dv{p^{(M)}(t)}{t} &= -\kappa p'(t) + O(\epsilon),
&
\dv{x^{(M)}(t)}{t} &= O(\epsilon).
\end{align}
These are the desired equations if one can ignore the $O(\epsilon)$
terms, the magnitude of which is determined by the error in
$\check\theta$ from the acquisition step as well as the
$O(1/\sqrt{M})$ quantum uncertainties in $J^{(M)}(t)/M$. Recall that
the covariance matrix of $X^{(M)}$ is
$\Re\Gamma[\delta(\theta)] \approx \Re \Gamma[\delta(\check\theta)]$,
so the covariance matrix of
$Y^{(M)} = \mqty(q^{(M)} & p^{(M)} &x^{(M)})^\top$ is
$L^{-1}(\Re \Gamma) L^{-\top}$, which does not scale with $M$. As long
as the $O(\epsilon)$ terms are much smaller than the normalized
uncertainties in $Y^{(M)}$ indicated by the covariance matrix, the
former can be ignored, and the linearized equations can be used to
propagate both the mean and the covariance matrix to show that the
spin statistics are transduced to the ancilla quadratures. This result
is the same as that derived from the QCLT-based approximation, as it
should be. Similar linearized equations of motion are widely used
  in the study of interactions between light and atomic ensembles
  \cite{hammerer10}.

\clearpage

\section*{\label{app_dyne}Supplemental material for \\
``Approaching the ultimate limit of quantum multiparameter estimation
by many-body physics''}

Let the quadratures of the $r$ additional ancilla modes for the
general-dyne measurement be
$\bigoplus_{j=1}^r \mqty(\tilde q_j& \tilde p_j)^\top$.  Assume that
their state is a zero-mean Gaussian state. Define 
\begin{align}
Z &\equiv \Bk{\bigoplus_{j=1}^r \mqty(\tilde q_j\\ -\tilde p_j)} \oplus
0,
\end{align}
where the last $0$ denotes the zero column vector with $n-2r$
entries. $Y^{(\infty)} + Z$ is then a vector of $n$ compatible
observables. The general-dyne measurement depicted in
Fig.~\ref{manybody_scheme3}(a) essentially measures the first $2r$
entries of $Y^{(\infty)} + Z$, while the homodyne measurement depicted
in Fig.~\ref{manybody_scheme3}(b) measures the rest. After multiplying
the column vector of outcomes by the $L$ matrix in Eq.~(\ref{L}), the
asymptotic normalized error becomes
\begin{align}
N\error_N \to \trace\bk{W\Re\Gamma} + \trace\bk{W L G L^\top},
\label{dyne_error}
\end{align}
where $\Re\Gamma$ is the covariance matrix of $X = LY$ and $G$ is the
covariance matrix of $Z$ (both in the sense of Wigner). For this error
to achieve the Holevo--Nagaoka bound, one should set
\cite[Sec.~3.2]{demkowicz20}
\begin{align}
\sqrt{W}L G L^\top \sqrt{W}&= \abs{\sqrt{W}\Im\Gamma\sqrt{W}},
\label{G}
\end{align}
so that the last term of Eq.~(\ref{dyne_error}) matches the last term
of the Holevo--Nagaoka bound in Eq.~(\ref{HN2}). Since $L$ is
invertible and $W$ is assumed to be strictly positive-definite, a
finite solution for $G$ always exists. Moreover, this $G$ is physical,
as demonstrated by Theorem~\ref{thm_G} in the following.  This theorem
is stated without proof in \cite[Sec.~3.2]{demkowicz20}; I provide an
explicit proof here for completeness. I need a simple lemma first.
\begin{lemma}
\label{lem_pos}
Let $A$ be a real and antisymmetric matrix and
$\abs{A} \equiv \sqrt{A^\top A}$.  Then the complex matrix
$\abs{A} + i A$ is positive-semidefinite.
\end{lemma}
\begin{proof}
Since $A$ is real and antisymmetric, there exists a real orthogonal
matrix $Q$ and positive constants $\{\xi_j\}$ such that
\begin{align}
A &= Q \BK{\Bk{\bigoplus_{j} \mqty(0 & \xi_j\\ -\xi_j & 0)} \oplus 0}Q^\top,
\\
\abs{A}&= Q \BK{\Bk{\bigoplus_j \mqty(\xi_j & 0\\ 0 & \xi_j)} \oplus 0}Q^\top,
\\
\abs{A} + i A &= 
Q \BK{\Bk{\bigoplus_j \xi_j B} \oplus 0}Q^\top,
\label{Acomplex}
\\
B &\equiv \mqty(1 & i\\ -i & 1).
\end{align}
Since $B \ge 0$, Eq.~(\ref{Acomplex}) is positive-semidefinite.
\end{proof}
\begin{theorem}
\label{thm_G}
Define 
\begin{align}
J &\equiv  \Bk{\bigoplus_{j=1}^r \mqty(0 & 1 \\ -1 & 0) }\oplus 0,
\end{align}
such that
\begin{align}
\Bk{Z_j,Z_k} &= -i J_{jk}.
\end{align}
The covariance matrix $G$ of $Z$ according to Eq.~(\ref{G}) is
physical in the sense of satsifying the uncertainty principle
\cite[Theorem~5.5.1]{holevo_aspect}
\begin{align}
G + \frac{i}{2} J  &\ge 0.
\label{unc}
\end{align}
\end{theorem}
\begin{proof}
Eqs.~(\ref{L}) and (\ref{Ycompat}) imply that
\begin{align}
\trace\bk{\Bk{Y_j,Y_k}\rho} &= i J_{jk}=
2i \bk{L^{-1} \Im\Gamma L^{-\top} }_{jk},
\label{Ycom}
\end{align}
leading to 
\begin{align}
\sqrt{W}L G L^\top \sqrt{W}= \frac{1}{2} \abs{\sqrt{W}L J L^\top \sqrt{W}},
\\
G = \frac{1}{2} \bk{\sqrt{W}L}^{-1} \abs{\sqrt{W}L J L^\top \sqrt{W}}
\bk{\sqrt{W}L}^{-\top}.
\end{align}
Then Eq.~(\ref{unc}) holds if and only if 
\begin{align}
\abs{A}+ i A &\ge 0,
&
A &= \sqrt{W}L J L^\top \sqrt{W}.
\label{Apos}
\end{align}
Since $A$ here is real and antisymmetric, Lemma~\ref{lem_pos} implies
Eq.~(\ref{Apos}) and thus Eq.~(\ref{unc}).
\end{proof}

In general, one needs a preliminary estimate of $\Im\Gamma$ to find
$L$ for the transformation $X = LY$ as well as $G$ to prepare the
general-dyne ancilla state.  In the special case $W = I$, however, $G$
is simple.
\begin{theorem}
\label{thm_G2}
  Assuming $W = I$, Eq.~(\ref{G}) holds if the general-dyne
  ancillas are in the vacuum state, viz.,
\begin{align}
G &= 
\Bk{\bigoplus_{j=1}^r \mqty(1/2 & 0\\ 0 & 1/2) }\oplus 0.
\label{vacuum}
\end{align}
In other words, the general-dyne measurement is simply heterodyne.
\end{theorem}
\begin{proof}
Define
\begin{align}
E &\equiv  
\Bk{\bigoplus_{j=1}^r \mqty(\sqrt{\nu_j} & 0\\ 0 & \sqrt{\nu_j}) }\oplus F,
\label{E}
\end{align}
where $F$ is any invertible $(n-2r)\times(n-2r)$ matrix, so that
$EJE^\top$ coincides with the right-hand side of Eq.~(\ref{ImG}).
Eq.~(\ref{ImG}) becomes
\begin{align}
\Im\Gamma &= O E J E^\top O^\top.
\label{ImG2}
\end{align}
The right-hand side of Eq.~(\ref{G}) becomes
\begin{align}
\abs{\Im\Gamma} &= O \abs{EJE^\top} O^\top,
\label{absImG}
\\
\abs{EJE^\top}
&= \Bk{\bigoplus_{j=1}^r \mqty(\nu_j & 0\\ 0 & \nu_j) }\oplus 0.
\label{absEJE}
\end{align}
On the other hand, Eqs.~(\ref{Ycom})  and (\ref{ImG2}) lead to
\begin{align}
L &= \sqrt{2} O E.
\end{align}
The left-hand side of Eq.~(\ref{G}) hence becomes
\begin{align}
L G L^\top &= 2 O E G E^\top O^\top
= O \abs{E J E^\top } O^\top = \abs{\Im\Gamma},
\end{align}
where the second equality comes from
$2 E G E^\top = \abs{EJE^\top}$ via Eqs.~(\ref{vacuum}),
(\ref{E}), and (\ref{absEJE}) and the last equality comes from
Eq.~(\ref{absImG}).
\end{proof}

\end{document}